\documentclass[11pt]{article}
\usepackage {amsmath, amssymb,amsthm, graphicx, mathrsfs}
\usepackage{amsmath, amssymb, amsthm, authblk, bigints}
\theoremstyle {plain}
\newtheorem {Thrm}{Theorem}[section]
\newtheorem {Lem}{Lemma}[section]
\newtheorem {Corol}{Corollary}[section]
\newtheorem {Prop}{Proposition}[section]
\theoremstyle {definition}

\newtheorem{Rem}{Remark}

\newtheorem {Def}{Definition}[section]
\numberwithin {equation}{section}
\newcommand{\bft}{\textbf}
\newcommand{\bsl}{\mathbf}
\newcommand{\ep}{\varepsilon}
\newcommand {\del}{\nabla}
\newcommand{\curl}{\textnormal {\bft {curl}}}
\newcommand{\dive}{\textnormal{div}}
\newcommand{\torus}{\mathbb {T}}

\newcommand{\uep}{\bsl {u}^\ep}

\newcommand{\dx}{\mathrm {d}\bsl {x}}
\newcommand{\dy}{\mathrm {d}\bsl {y}}

\newcommand{\nbb}{\mathbb {N}}

\newcommand{\xy}{(\bsl {x},\bsl {y})}

\newcommand{\bphi}{ \boldsymbol {\varphi}}

\newcommand{\avalph}{\langle \alpha \rangle}
\newcommand{\avalphin}{\langle \alpha^{-1}\rangle ^{-1}}

\newcommand{\scalemath}[2]{\scalebox{#1}{\mbox{\ensuremath{\displaystyle #2}}}}

\def\XXint#1#2#3{{\setbox0=\hbox{$#1{#2#3}{\int}$ }
\vcenter{\hbox{$#2#3$ }}\kern-.6\wd0}}

\makeatletter
\def\blfootnote{\xdef\@thefnmark{}\@footnotetext}
\makeatother

\begin{document}

\title{\sc  Full two-scale asymptotic expansion and higher-order constitutive laws in the homogenisation of the system of Maxwell equations}
\author[1]{Kirill D. Cherednichenko}
\author[2]{James A. Evans}
\affil[1]{Department of Mathematical Sciences, University of Bath, Claverton Down, Bath, BA2 7AY, UK}
\affil[2]{Cardiff School of Mathematics, Cardiff University, Senghennydd Road, CF24 4AG, UK}

\maketitle

\begin{abstract}
For the system of Maxwell equations of electromagnetism in an $l$-periodic composite medium of overall size $L$ ($0<l<L<\infty$), in the low-frequency quasistatic approximation, we develop an electromagnetic version of strain-gradient theories, where the magnetic field is not a function of the magnetic induction alone but also of its spatial gradients, and the electric field depends not only on the displacement but also on displacement gradients. 
Following the work (Smyshlyaev, V.P., Cherednichenko, K.D., 2000. On rigorous derivation of strain gradient effects in the overall behaviour of periodic heterogeneous media, {\it J. Mech. Phys. Solids} {\bf 48}, 1325--1357), we develop a combination of variational and asymptotic approaches to the multiscale analysis of the Maxwell system. We provide rigorous convergence estimates of higher order of smallness with respect to the inverse of the ``scale separation parameter" $L/l.$ Using a special ``ensemble averaging'' procedure for a family of periodic problems, we derive an infinite-order version of the standard
homogenised operator of second order.
  
\end{abstract}




	\section {Introduction}
	Heterogeneous media with microstructure (``composites''), which in the simplest case involve two length-scales: ``macroscopic'' $L$ and ``microscopic'' $l,$ so that $L\gg l,$ have been a focus of attention across engineering and physical sciences 
	since the advent of quantum mechanics and the subsequent development of the ideas of ``upscaling'' and ``course-graining'' in the quantitative description of material properties.  It is widely known that under the assumptions of periodicity  (where both $L,$ $l$ are periods in the variation of material properties) and strict separation of scales, {\it i.e.} in the regime when $l/L\to0$ is a suitable approximation, the behaviour of composites is fully determined by the solution of the so-called ``cell problems'' on the period of the composite.  
	
It has also been noticed in the context of elastic solids (see {\it e.g.} \cite{FH}) that when the ratio  $l/L$ is ``small but not too small'', additional terms need to be taken into account in order to obtain an accurate picture of the overall response of the composite. Mathematically speaking, the solution to the cell problem does not suffice when the length-scales involved are not well separated from each other, and a more sophisticated ``higher-order'' averaging framework has to be invoked, as was shown in \cite{bib13}, \cite{CherSm_2004}.  

In the present work, we develop the higher-order approach for the analysis of the classical system of Maxwell equations. 
Our motivation for doing so stems from the increasing interest to ``metamaterials'' (see {\it e.g.} \cite{RamakrishnaGrzegorczyk} for an overview) and from the indication (see \cite{Zhikov2000}, \cite{CherCooper}) that the quantitative  description of the effect of length-scale interactions on the overall behaviour of composites is key to developing new ways of manufacturing metamaterials. 

The centrepiece of the homogenisation theory for second-order elliptic equations (see {\it e.g.}  
\cite {bib39}, 
\cite {bib25}, 
\cite {bib5})
is the derivation of the homogenised equation that captures the ``effective'' behaviour of the original problem by ``averaging out"  small-scale oscillations. 
By contrast, the main products of our approach are the ``infinite-order homogenised equation'', see (\ref{hom1}), (\ref{EL4}), involving  
a sum of higher-order gradients of the classical homogenised solution, with coefficients given by increasing powers of the scale parameter 
$\varepsilon:=l/L,$
from which the classical homogenised equation is obtained as the  truncation of lowest order, and the ``homogenised equation of higher order" obtained as the Euler-Lagrange equation for the ``higher-order homogenised variational problem" (\ref{UK1}).

The use of asymptotic expansions in understanding 
size effects in periodic media 
was proposed in 
\cite {bib114}, 
\cite {bib112}, 
where 
the solution to linearised elasticity equations in a periodic composite medium was sought in the form of a two-scale asymptotic series expansion whose terms depend on the macroscopic $\bft{x}$ and microscopic $\bft{x}/\varepsilon$ variable.
Using a special averaging procedure, 
a set of 
``higher-order stress-strain relations'' for such media  is derived in
\cite {bib114} 
and
\cite {bib112},  where the stress is a function of not only strain but also of gradients of strain. 
The advantage of the asymptotic approach in deriving such constitutive laws is the availability of rigorous error estimates, for small values of  $\ep,$ for the difference between the actual strain and stress and their higher-order versions.  
However, this approach also has two drawbacks from the perspective 
of numerical implementation
the higher-order expansions 
are only expected to be accurate if $\ep$ is sufficiently small 
and 
may  lead non-elliptic differential equations, as shown in  
\cite {bib111}.

Following 
\cite {bib13}, 
who study the equations of elasticity, we develop a higher-order homogenisation framework for the system of Maxwell equations,
via a combination of asymptotic and variational calculi. In the first part of our analysis, Section \ref{sec112},
we provide a higher-order extension of the classical two-scale asymptotic approach, including the higher-order version of the standard convergence estimates.
In the second part (Section \ref{sec12}), we use a version of ``variational asymptotics"  to derive higher-order homogenised equations that are elliptic by construction. This is achieved by using a set of trial fields suggested by the asymptotic analysis in a family of variational formulations associated with the problem for the original composite medium. 
Furthermore, we show that the corresponding higher-order variational solution is ``close",  in a certain variational sense (Proposition \ref{prop3}), to the solution of the original problem. 
The main idea, which we adopt from  \cite {bib13}, is to cancel the effect of the rapid oscillations in higher-order terms via the ``ensemble averaging'' of a family of problems obtained
by shifting the fast variable of the original problem. 

By analogy with \cite {bib13}, the asymptotic and variational approaches result in equivalent ``infinite-order'' homogenised equation (Sections \ref{sec113} and \ref{sec125}), however the proof of this fact in the Maxwell case requires a special tensor symmetrisation procedure, which we carry out in Section \ref{sec126}.

Finally, in Section \ref{sec13}, we derive the infinite-order effective constitutive relations between the magnetic field and induction and between electric field and displacement. In analogy with the elasticity case, the magnetic field depends not only on the magnetic induction but also on its spatial derivatives, with coefficients that become more significant for larger values of $\varepsilon.$ This leads to the generalisation of the notion of macroscopic magnetic permeability, in the form of an ``effective permeability operator'', whose inverse is a differential operator of infinite order, see Section \ref{sec133}. 


  
In the main body of our work (Sections \ref {sec11}--\ref{sec12}) we focus on the pair of equations governing the behaviour of the electric field when the electric permittivity is constant. Similar analysis, however, applies to the pair of equations for the magnetic field, where electric permittivity rapidly oscillates and magnetic permeability is constant, subject to additional considerations due to the presence of an oscillatory current density term, as 
 discussed in Section \ref {sec134}. By analogy with the case of the magnetic field and induction, we derive 
 an infinite-order ``effective permittivity operator'', which links the overall electric field and displacement.

	\section {Formulation of the problem}\label {sec11}


			In what follows we study the vector equation
				\begin {equation}\label {firsteq}
					\textbf {curl}\Big\{A\Big (\frac {\bsl {x}}{\varepsilon}\Big )\textbf {curl}\ \bsl {u}^{\varepsilon} ( \bsl {x} ) \Big\}=\bsl {f}(\bsl {x}),\ \ \ \ \ \ \bsl {x}\in \mathbb {T}:=[0,T]^3,\ \ 
					\varepsilon, T>0,\ \ T/\varepsilon\in{\mathbb N},
				\end {equation}
			where the coefficient matrix $A$ is assumed to be measurable, $Q$-periodic, symmetric:
			\[
			A_{ij}(\bsl {y})=A_{ji}(\bsl {y})\ \ \ \ \forall\,{\bsl y}\in Q:=[0,1)^3,\ \ i,j=1,2,3,
			\]
			bounded and uniformly elliptic: 
	
	\[
	\exists\,\nu>0:\ \ \nu | \boldsymbol {\xi} |^2\leq A_{ij}(\bsl {y})\xi_i\xi_j \leq \nu^{-1}|\boldsymbol {\xi} |^2\ \ \ \ \forall\,\boldsymbol {\xi} \in \mathbb {R}^3,\,{\bsl y}\in Q.
	\]
	Throughout the paper, we use the notation $\hat{A}^\varepsilon(\cdot):=A(\cdot/\varepsilon)$ for matrix functions $A$ with the above properties.
	We take right-hand sides ${\bsl f}$
	that are infinitely smooth in ${\mathbb R}^3,$ ${\mathbb T}$-periodic, divergence-free and have zero mean. 
The equation (\ref {firsteq}) describes the behaviour of the electric component of an electromagnetic field in the quasistatic approximation (see 
		Section \ref{sec133}). Here the quantities $A,$ $\uep,$ and $\bsl {f}$ represent the inverse of the magnetic permeability $\hat {\mu}$, the electric field  $\bsl {E}_1^\ep$ and the current density $-\bsl {J}_0,$ respectively,  at each point $\bsl{x}\in{\mathbb R}^3 ({\rm mod}\,{\mathbb T}).$ In terms of the ``size parameters'' $l, L$ mentioned in the previous section, we assume that $\varepsilon/T=l/L.$ 

			Define the space $H^1_{\curl}(\torus)$ to be the closure of the set $\bigl[C^\infty_{\rm per}(\torus)\bigr]^3$ of infinitely differentiable ${\mathbb T}$-periodic functions with respect to the norm
				\begin {equation}
				\label {norm1}
					\| \bsl {u}\|_{H^1_{\curl}(\torus)}=\| \bsl {u}\|_{[L^2(\torus)]^3}+\| \curl\,\bsl {u}\|_{[L^2(\torus)]^3}.
				\end {equation}
				We shall study weak solutions $\bsl {u}^{\ep}$ to (\ref{firsteq})
			in the space
				\begin{equation}
				{\mathcal X}(\torus ):=\big \{ \bsl {u}\in \bigl[L^2(\mathbb {T})\bigr]^3\,\big|\ 
				\dive\,\bsl {u}=0,\ \langle \bsl {u}\rangle_\torus =\bsl {0}\big \}\cap H^1_{\curl}(\mathbb {T}),
				\label{spaceX}
				\end{equation}
			 so that the identity
			 \begin {equation*}
		\int_{\mathbb {T}}\hat{A}^\varepsilon
		\textbf {curl}\,\bsl {u}^{\varepsilon}\cdot \textbf {curl}\,\boldsymbol {\varphi}
		=\int_{\mathbb {T}}\bsl {f}
		\cdot \boldsymbol {\varphi}
		\ \ \ \ \ \forall\, \boldsymbol {\varphi} \in\bigl[C^\infty_{\rm per}(\mathbb{T})\bigr]^3
	          	\end {equation*}
		holds. The divergence ${\rm div}$ in (\ref{spaceX}) is understood in the sense of distributions.
		Henceforth, we use the notation $\langle\ \cdot\ \rangle_{\mathbb T}$, $\langle\ \cdot\ \rangle$ for 
		averages (mean values) over ${\mathbb T}$ and $Q,$ respectively: 
		\[
		\langle f\rangle_{\mathbb T}:=\frac {1}{|{\mathbb T}|}\int_{\mathbb T}f(\bsl {x})\,\dx,\ \ \ \langle g\rangle:=\frac {1}{|Q|}\int_Qg(\bsl {y})\,\dy=\int_Qg(\bsl {y})\,\dy,
		\]
and the averages of vector quantities are taken component-wise.
 Note that ${\mathcal X}(\torus)$ equipped with the norm (\ref{norm1}) is a Sobolev space. For all $\ep >0$, the equation (\ref {firsteq}) is well posed in $X(\torus ),$ as shown next. 
				\begin {Thrm}\label {unique1}
					For all ${\bsl f}\in\bigl[C^\infty_{\rm per}(\mathbb {T})\bigr]^3\cap {\mathcal X}(\torus)$, there exists a unique solution 
					${\bsl u}^{\varepsilon}\in {\mathcal X}(\mathbb {T})$  to (\ref{firsteq}).
				\end {Thrm}
				\begin {proof}
					Define a bilinear form ${\mathfrak b}$
					on 
					${\mathcal X}(\torus)$
					by the formula
						$${\mathfrak b}(\bsl {u}, \boldsymbol {\varphi})=\int_{\mathbb {T}}
						\hat{A}^\varepsilon
						\textbf {curl}\ \bsl {u}\cdot \textbf {curl}\ \boldsymbol {\varphi},
						\ \ \ \ \bsl {u}, \boldsymbol {\varphi}\in {\mathcal X}(\torus).
						$$
					A unique solution to the problem 
						$${\mathfrak b}(\bsl {u}^{\ep}, \boldsymbol {\varphi})=\int_{\mathbb T}\bsl {f}\cdot \boldsymbol {\varphi}
						\ \ \ \ \forall 
						 \boldsymbol {\varphi}\in {\mathcal X}(\torus)$$
					exists by the Lax-Milgram Lemma (see {\it e.g.} 
					\cite {bib24}). 
Indeed, coercivity of the form ${\mathfrak b}$ follows from ellipticity of the matrix $A:$ 
						$${\mathfrak b}(\bsl {u},\bsl {u})
						\geq \int_{\mathbb {T}}\nu |\bft {curl}\,\bsl {u}|^2
						\geq \frac {1}{2}\nu \int_{\mathbb {T}} |\bft {curl}\,\bsl {u}|^2
						+\frac {1}{2}\nu \int_{\mathbb {T}}|\bft {curl}\,\bsl {u}|^2
						\geq \frac {1}{2}\nu \int_{\mathbb {T}} |\bft {curl}\,\bsl {u}|^2
						+\frac {1}{2}\nu C_1\int_{\mathbb {T}}|\bsl {u}|^2
						\geq C_2\| \bsl {u}\|^2_{H^1_{\curl}(\mathbb {T})}$$
						for some $C_1, C_2>0,$ where we have used (\ref{Maxwell_inequality}), see Appendix A,
						combined with the fact that ${\bsl u}$ is divergence-free.
					Continuity of the form ${\mathfrak b}$ follows from the Cauchy-Schwarz inequality and the assumption that $A$ is bounded:
						$$
						{\mathfrak b}(\bsl {u}, \boldsymbol {\varphi}) 
						\leq \Big (\int_{\mathbb {T}}|A\,\bft {curl}\,\bsl {u}|^2
						\Big )^{1/2} 
						\Big (\int_{\mathbb {T}}|\bft {curl}\,\boldsymbol {\varphi}|^2
						\Big )^{1/2}\leq\widehat {C}\nu^{-1} \|\bsl {u}\|_{H^1_{\curl}(\mathbb {T})} \| \boldsymbol {\varphi}\|_{H^1_{\curl}(\mathbb {T})}
						$$
						for some $\widehat{C}>0.$
				\end {proof}


		\section {Asymptotic expansion of the solution to  (\ref {firsteq})}\label {sec112}
			 We seek a solution of (\ref{firsteq}) 
			 in the form of a two-scale power series:
				\begin {equation}
					\bsl {u}^{\varepsilon}(\bsl {x})=\sum_{j=0}^{\infty}\varepsilon^j\bsl {u}_j\Big ( \bsl {x}, \frac {\bsl {x}}{\ep}\Big ), \ \ \ \ \ \ \bsl{x}\in{\mathbb T}.
					\label{main_series}
				\end {equation}
			Substituting (\ref{main_series}) into equation (\ref{firsteq})
			yields a more specific form for the coefficients $\bsl {u}_j,$ namely:
				\begin {multline}\label {A2}
					\bsl {u}^{\varepsilon}(\bsl {x})=\bsl {v}(\bsl {x},\ep)+\sum_{j=1}^{\infty}\ep^j\Big \{ \del_\bsl{y}\big ({\mathscr K}^{(j)}(\bsl {y})\del_{\bsl x}^j\bsl {v}(\bsl {x},\ep)\big )+ \del_{\bsl x}\big ({\mathscr K}^{(j-1)}(\bsl {y})\del_{\bsl x}^{j-1}\bsl {v}(\bsl {x},\ep)\big )\\[-0.8em]
					+{\mathscr N}^{(j)}(\bsl {y})\del_{\bsl x}^{j-1}\bft {curl}_{\bsl x}\bsl {v}(\bsl {x},\ep)\Big \}\Bigr\vert_{\bsl{y}={\bsl x}/{\varepsilon}},
				\end {multline}
			where we set ${\mathscr K}^{(0)}=0.$ 
			Henceforth, we denote by $\nabla_{\bsl x}^j,$ $\nabla_\bsl{y}^j,$ $j=0,1,2,...,$ the tensors of derivatives of order $j$ with respect to the variables 
			${\bsl x},$ ${\bsl y},$ 
			often omitting the lower index when it is clear from the context. 
			The coefficients ${\mathscr K}^{(j)}$ and ${\mathscr N}^{(j)}$ are tensors of order $j+1,$ $j=0,1,\dots,$ whose components 
		belong to the spaces
				$H^2_{{\rm per}, 0}(Q)$ and 
				$H^1_{{\rm per}, 0}(Q),$ 
				respectively, where $H^l_{{\rm per}, 0}(Q),$ $l=1,2,$ denotes $H^l(Q)$-closure of the set of elements of $\bigl[C^\infty_{\rm per}(Q)\bigr]^3$ with zero mean. The smooth, $\torus$-periodic, divergence-free vector field ${\bsl v}(\cdot, \varepsilon)$ is 
			sought as a series in powers of 
			$\varepsilon:$
						\begin {equation}\label {asymptoticv}
							\bsl {v}(\bsl {x},\ep)=
							\sum_{k=0}^\infty \ep^k\bsl {v}_k(\bsl {x}),\ \ \ \ \ \ \bsl{x}\in{\mathbb T},
						\end {equation}
					where $\bsl {v}_k,$ $k=0,1,2,\dots,$ are $\ep$-independent. 
					
						\begin{Rem}
						 The tensor products in (\ref {A2}) are evaluated as follows:
								\[
								{\mathscr K}^{(j)}\del_{\bsl x}^j\bsl {v}={\mathscr K}^{(j)}_{i_1i_2\dots i_{j+1}}v_{i_{j+1},i_1\dots i_j},\ \ \ \ \ \
								{\mathscr N}^{(j)}\del_{\bsl x}^{j-1}\bft {curl}_{\bsl x}\bsl {v}={\mathscr N}^{(j)}_{i_1i_2\dots i_{j+1}}(\bft {curl}_{\bsl x}\bsl {v})_{i_{j+1},i_2\dots i_j},
								\]
							for $i_k\in\{1,2,3\},\ k=1,2,\dots ,j+1$, where summation is carried out for repeated indices, and the comma denotes differentiation with respect to the indices following the comma. 
				The curl and divergence of a tensor ${\mathscr N}^{(j)}$ of order $j+1$  are tensors of order $j+1$  and $j,$ respectively, given by
								\[
								\bigl(\curl \,{\mathscr N}^{(j)}\bigr)_{i_1i_2\dots i_{j+1}}=\epsilon_{i_1st}{\mathscr N}^{(j)}_{ti_2\dots i_{j+1},s}\ \ \ \ \ \bigl(\dive\,{\mathscr N}^{(j)}\bigr)_{i_1i_2\dots i_{j}}={\mathscr N}^{(j)}_{si_1i_2\dots i_{j},s}.
								\]
				\end {Rem}
We formally substitute (\ref {A2}) into the equation (\ref {firsteq}), treating the ``slow'' variables $(x_1,x_2,x_3)=:\bsl {x}$ and the ``fast'' variables $(y_1,y_2,y_3)=:\bsl {y}$ independently, so the ``full'' gradient and curl operators are evaluated according to the rules
$\del =\del_{\bsl x}+\ep^{-1}\del_\bsl{y},$ and  $\bft {curl}=\bft {curl}_{\bsl x}+\ep^{-1}\bft {curl}_\bsl{y}.$
			Making use of the identities 
				$\curl_{\bsl x}\del_{\bsl x} (\cdot )=\curl_\bsl{y}\del_\bsl{y} (\cdot )=\bsl {0}$  and $\curl_{\bsl x}\del_\bsl{y}(\cdot)
				=-\curl_\bsl{y}\del_{\bsl x}(\cdot ),$ 
				we formally write
				\vskip-0.8cm
				\begin {multline}\label {mess1}
					\curl_{\bsl x}\{A\,\curl_{\bsl x}\bsl {v}\}+\ep^{-1}\curl_\bsl{y}\{A\,\curl_{\bsl x}\bsl {v}\}+\sum_{j=1}^{\infty}\ep^j\curl_{\bsl x}\bigl\{A\,\curl_{\bsl x}\bigl({\mathscr N}^{(j)}\del_{\bsl x}^{j-1}\curl_{\bsl x}\bsl {v}\bigr)\bigr\}\\[-0.6em] +\sum_{j=1}^{\infty}\ep^{j-1}\Bigl(\curl_\bsl{y}\bigl\{A\,\curl_{\bsl x}\bigl({\mathscr N}^{(j)}\del_{\bsl x}^{j-1}\curl_{\bsl x}\bsl {v}\bigr)\bigr\}+\curl_{\bsl x}\bigl\{A\,\curl_\bsl{y}\bigl({\mathscr N}^{(j)}\del_{\bsl x}^{j-1}\curl_{\bsl x}\bsl {v}\bigr)\bigr\}\Bigr)\\[-0.7em] +\sum_{j=1}^{\infty}\ep^{j-2}\curl_\bsl{y}\bigl\{A\,\curl_\bsl{y}\bigl({\mathscr N}^{(j)}\del_{\bsl x}^{j-1}\curl_{\bsl x}\bsl {v}\bigr)\bigr\}=\bsl {f}.
				\end {multline}
			Equating the terms in (\ref{mess1}) corresponding to individual powers of $\varepsilon,$ we obtain: 
				\begin {equation}
				\label {L2}
					\curl_\bsl{y}\bigl\{A\,\curl_\bsl{y}\bigl({\mathscr N}^{(1)}\curl_{\bsl x}\bsl {v}\bigr)\bigr\}=-\curl_\bsl{y}\{A\,\curl_{\bsl x}\bsl {v}\},
				\end{equation}
\vskip-0.8cm
	\begin {multline}
				\label {L3}
				\curl_\bsl{y}\bigl\{A\,\curl_\bsl{y}({\mathscr N}^{(2)}\del_{\bsl x}\curl_{\bsl x}\bsl {v})\bigr\}=\bsl{f}-\curl_{\bsl x}\{A\,\curl_{\bsl x}\bsl {v}\}\\[0.4em] -\curl_\bsl{y}\bigl\{A\,\curl_{\bsl x}({\mathscr N}^{(1)}\curl_{\bsl x}\bsl {v})\bigr\}
				-\curl_{\bsl x}\bigl\{A\,\curl_\bsl{y}({\mathscr N}^{(1)}\curl_{\bsl x}\bsl {v})\bigr\}.
	\end {multline}
			Requiring that (\ref {L2}) be satisfied for all admissible vector fields $\bsl {v}$ implies
				\begin{equation}
				\curl
				\bigl\{A\,\curl\,
				{\mathscr N}^{(1)}\bigr\}=-\curl\,
				A,
				\label{N1_eq}
				\end{equation}
				which is understood in the weak sense, {\it i.e.}
				\begin{equation}
				\bigl\langle\bigl(A\,\curl\,{\mathscr N}^{(1)}+A\bigr)\curl\,\phi\bigr\rangle=0\ \ \ \ \forall\phi\in\bigl[C^\infty_{\rm per}(Q)\bigr]^3.
				\label{star}
				\end{equation}
			The matrix ${\mathscr N}^{(1)}$ is determined uniquely under the conditions that it is
			 $Q$-periodic and has zero average.  
			 The condition of solvability of (\ref {L3}), viewed as an equation for ${\mathscr N}^{(2)},$ is
			 the usual homogenised equation ({\it cf.} \cite{bib5})
				\begin{equation}
				\curl_{\bsl x}\bigl\{\hat {h}^{(2)}\curl_{\bsl x}\bsl {v}\bigr\}=\bsl {f},
\ \ \ \ \bsl{v}\in{\mathcal X}(\torus),\ \ \ \ \ \ \ \ \ \hat {h}^{(2)}:=\bigl\langle A\bigl( \curl\ {\mathscr N}^{(1)}+I\bigr)\bigr\rangle.
\label{hom_eq}
\end{equation}
				\begin {Rem}
				\label{Sec3_remark}
					In the case of a ``laminate'', $A(\bsl {y})=\alpha(y_2)I,$ the matrix ${\mathscr N}^{(1)}$ is found to have the form
						\begin {equation}\label {mat1}
							{\mathscr N}^{(1)}=\begin {pmatrix}
							0 & 0 & -N \\
							0 & 0 & 0 \\
							N & 0 & 0 \end {pmatrix},
						\end {equation}
					where $N=N(y_2)$ satisfies $-(\alpha N')'=\alpha ',$ subject to the conditions that it is $Q$-periodic and has zero mean.
					Direct calculation shows that 
						\begin {equation}\label {Nexact}
							N(y_2)=\int_0^{y_2}\bigl(\avalphin \alpha^{-1}(t)-1\bigr)\mathrm {d}t-\int_0^1\int_0^{y_2}\bigl(\avalphin \alpha^{-1}(t)-1\bigr)\mathrm{d}t\mathrm{d}y_2.
						\end {equation}
				\end {Rem}

		\subsection {Infinite-order homogenised equation: asymptotic approach}\label {sec113}
			
			Denote by $\bsl{H}_j=\bsl {H}_j(\bsl {x},\bsl {y})$ the coefficient in front of $\varepsilon^j,$ $j=0,1,2,...,$
			in the expansion (\ref {mess1}): 
				\begin {multline}
				\label {HHH}
					\bsl {H}_j:= \curl_{\bsl x}\bigl\{A\,\curl_{\bsl x}\bigl({\mathscr N}^{(j)}\del_{\bsl x}^{j-1}\curl_{\bsl x}\bsl {v}\bigr)\bigr\}\\[0.5em] +\curl_\bsl{y}\bigl\{A\,\curl_{\bsl x}\bigl({\mathscr N}^{(j+1)}\del_{\bsl x}^{j}\curl_{\bsl x}\bsl {v}\bigr)\bigr\}
					+\curl_{\bsl x}A\bigl\{\,\curl_\bsl{y}\bigl({\mathscr N}^{(j+1)}\del_{\bsl x}^{j}\curl_{\bsl x}\bsl {v}\bigr)\bigr\}\\[0.5em] +\curl_\bsl{y}\bigl\{A\,\curl_\bsl{y}\bigl({\mathscr N}^{(j+2)}\del_{\bsl x}^{j+1}\curl_{\bsl x}\bsl {v}\bigr)\bigr\}, \ \ \ j=0,1, 2, \dots .
				\end {multline}
			We aim to write $\bsl {H}_j\xy$ in the form $h^{(j+2)}(\bsl {y})\del_{\bsl x}^{j+1}\curl_{\bsl x}\bsl {v}(\bsl{x}),$ by commuting all $\bsl {x}$-derivatives through to the right-hand sides of the tensors ${\mathscr N}^{(j)}(\bsl {y})$ in expressions 
			(\ref {HHH}). To this end, we introduce tensors ${\mathscr M}^{(j)}(\bsl {y})$ and ${\mathscr L}^{(j)}(\bsl {y})$ of order 
			$j+1,$ such that the following operator identities hold: 
				\begin {equation}
				\label {MMM}
					{\mathscr M}^{(1)}=I,\ \ \ \ {\mathscr M}^{(j)}(\bsl {y})\del_{\bsl x}=\curl_{\bsl x}{\mathscr N}^{(j-1)}(\bsl {y}),\ \ \ j=2,3,...,
				\end {equation}
				\begin {equation}
				\label {LLL}
					{\mathscr L}^{(j)}(\bsl {y})\del_{\bsl x}=\curl_{\bsl x}\Bigl(A(\bsl{y})\big \{ \curl_\bsl{y}{\mathscr N}^{(j-1)}(\bsl {y})+{\mathscr M}^{(j-1)}(\bsl {y})\big \}\Bigr),\ \ \ j=2,3,...\ .
				\end {equation}
			The expression for $\bsl {H}_j$, $j=0,1,2,\dots $ is now rewritten as
				\begin{gather*}
				\bsl {H}_j(\bsl {x},\bsl {y})=h^{(j+2)}(\bsl {y})\del_{\bsl x}^{j+1}\curl_{\bsl x}\bsl {v}(\bsl{x}),\ \ \ \ \ \ \ \ \ \ \ \ \ \ \ \ \ \ \ \ \ \ \ \ \ \ \ \ \ \ \ \ \ \ \ \ \ \ \ \ \ \ \ \ \ \ \ \ \ \ \ \ \ \ \ \ \ \ \ \ \ \ \\[0.4em]
				\ \ \ \ \ \ \ \ \ \ \ \ \ \ \ \ \ \ \ \ h^{(j+2)}(\bsl {y}):=\curl
				\bigl\{A(\bsl{y})\,\curl\,
				{\mathscr N}^{(j+2)}(\bsl {y})\bigr\}+\curl
				\bigl\{A(\bsl{y}){\mathscr M}^{(j+2)}(\bsl {y})\bigr\}+{\mathscr L}^{(j+2)}(\bsl {y}).
				\end{gather*}
			Summarising, the left-hand side of (\ref{firsteq}) takes the form
				\begin {equation*}
					\textbf{curl}\big\{\hat{A}^\varepsilon(\bsl{x})\textbf {curl}\,\bsl {u}^{\varepsilon}(\bsl {x})\big\}=\sum_{j=0}^\infty \ep^jh^{(j+2)}(\bsl {y})\del_{\bsl x}^{j+1}\curl_{\bsl x}\bsl {v}(\bsl {x})\bigr\vert_{\bsl{y}=\bsl{x}/\varepsilon}.
				\end {equation*}
			By analogy with the matrix $\hat{h}^{(2)},$ we require the tensors $h^{(j+2)},$ $j=0,1,2, ...$ to be independent of the fast variable $\bsl {y}.$
			The resulting system of recurrence relations
				\begin {equation}\label {hom-2}
					\curl
					\bigl\{A\,\curl\,
					{\mathscr N}^{(j+2)}\bigr\}=-\curl
					\bigl\{A\,{\mathscr M}^{(j+2)}\bigr\}-{\mathscr L}^{(j+2)}+h^{(j+2)},\ \ \ j=0,1,...,
				\end {equation}
			allows one to determine uniquely the tensors ${\mathscr N}^{(j+2)}, h^{(j+2)},\ j=0,1,\dots ,$ by virtue of the following statement.
				\begin {Lem}\label {lem113}
					Let $F^{(j)}$ be a tensor field 
					whose components are differentiable and $Q$-periodic. Furthermore, assume that $A$ is a positive definite $Q$-periodic matrix field. 
					Then in order that the equation 
						\begin {equation}\label {Nj}
							\curl
							\bigl\{A\,\curl\,
							{\mathscr N}^{(j)}\bigr\}=F^{(j)},
						\end {equation}
have a divergence-free, $Q$-periodic solution ${\mathscr N}^{(j)}$ whose elements have zero mean over $Q,$							
it is necessary and sufficient that						
						\begin{equation}
						\dive\,F^{(j)}=0,\ \ \ \bigl\langle F^{(j)}\bigr\rangle = 0.
						\label{RHS_conditions}
						\end{equation}
				\end {Lem}
				\begin {proof}
					The necessity follows by taking the divergence and average of both sides of equation (\ref {Nj}) and noting that both $A$ and ${\mathscr N}^{(j)}$ are 
					$Q$-periodic. 
					Conversely, due to the fact that the differential expression in the left-hand side of (\ref{Nj}) defines a self-adjoint operator in the space of divergence-free tensor fields with  zero mean, it suffices to note that the right-hand side of (\ref {Nj}) is $L^2(Q)$-orthogonal 
					to all elements of the 
					kernel of the left-hand side, namely\footnote{
					If $F^{(j)}$ and $G^{(j)}$ are two tensors of order $j$, then 
						$F^{(j)}\cdot G^{(j)}:=(F^{(j)})_{i_1i_2\dots i_j}(G^{(j)})_{i_1i_2\dots i_j}.$}
						\begin{equation}
						\bigl\langle F^{(j)}\cdot \del\varphi\bigr\rangle=0,\ \ \ \ \ \bigl\langle F^{(j)}\cdot{\mathscr C}\bigr\rangle=0,
						\label{orthog}
						\end{equation}
for all smooth order $j-1$ tensor fields $\varphi$ and order $j$ constant tensors ${\mathscr C}.$ The proof is concluded by noting that (\ref{orthog}) follows immediately from (\ref{RHS_conditions}).
				\end {proof}
			By Lemma \ref {lem113}, the equation (\ref{hom-2}) is solvable if and only if 
				$\bigl\langle -\curl
				(A {\mathscr M}^{(j+2)})-{\mathscr L}^{(j+2)}+h^{(j+2)}\bigr\rangle =0,$ 
				or equivalently $h^{(j+2)}= \bigl\langle {\mathscr L}^{(j+2)}\bigr\rangle.$
			It follows that
				$$ h^{(j+2)}\del
				^{j+1}\curl\,
				\bsl {v}(\bsl {x},\ep) =\bigl\langle {\mathscr L}^{(j+2)}\bigr\rangle\del
				^{j+1}\curl\,
				\bsl {v}(\bsl {x},\ep)
				=\curl
				\Bigl(\bigl\langle A
				\bigl\{\curl\,
				{\mathscr N}^{(j+1)}
				+{\mathscr M}^{(j+1)}
				\bigr\}\bigr\rangle\del
				^{j}\curl\,
				\bsl {v}(\bsl {x},\ep)\Bigr).$$
			The ``infinite-order homogenised equation" takes the form
				\begin {equation}
				\label {hom1}
					\curl \bigl\{\hat {h}^{(2)}\curl\,\bsl {v}(\bsl {x},\ep)\bigr\}+ \sum_{j=1}^{\infty}\ep^j\curl
					\bigl\{\hat {h}^{(j+2)}\del
					^j\curl\,
					\bsl {v}(\bsl {x},\ep)\bigr\}=\bsl {f}(\bsl {x}),\ \ \ \ {\bsl x}\in{\mathbb T},
				\end {equation}
			where 
				\begin {equation}\label {homco1}
					\hat {h}^{(j+2)}:=\Bigl\langle A\bigl\{\curl\,{\mathscr N}^{(j+1)}+{\mathscr M}^{(j+1)}\bigr\}\Bigr\rangle,\ \ \  j=0,1,2,...,
				\end {equation}
			are the tensors of ``higher-order homogenised coefficients''. They 
			are related to the tensors $h^{(j+2)}$ via the operator identity
					$\curl_{\bsl x}\hat {h}^{(j+2)}=h^{(j+2)}\del_{\bsl x}.$
				In the index notation, the formula (\ref{homco1}) reads 
				\[
				\hat{h}^{(j+2)}_{i_1...i_{j+2}}=\Bigl\langle A_{i_1s}\bigl\{\curl\,{\mathscr N}^{(j+1)}+{\mathscr M}^{(j+1)}\bigr\}_{si_2...i_{j+2}}\Bigr\rangle.
				\]
				\begin {Rem}
					For the example $A(\bsl {y})=\alpha(y_2)I$, we have
					\[
					\hat {h}^{(2)}=\Bigl \langle A\bigl\{ \curl\,{\mathscr N}^{(1)}+I\bigr \} \Bigr \rangle=\begin {pmatrix}
						\langle \alpha^{-1}\rangle^{-1} & 0 & 0 \\
						0 & \langle \alpha \rangle & 0 \\
						0 & 0 & \langle \alpha^{-1}\rangle^{-1}\end {pmatrix}.
					\]
				\end {Rem}
				\begin {Rem}
					Using index notation, the relations (\ref {MMM})--(\ref {LLL}) read
						\begin {equation}\label {Mij}
							\bigl({\mathscr M}^{(j+1)}\bigr)_{i_1\dots i_{j+2}}=\epsilon_{i_1i_2s}({\mathscr N}^{(j)})_{si_3i_4\dots i_{j+2}},\ \ \ j=1,2,\dots ,
						\end {equation}
						\begin {equation}\label {Lij}
							\bigl({\mathscr L}^{(j+1)}\bigr)_{i_1\dots i_{j+2}}=\epsilon_{i_1i_2s}A_{st}\big \{ \curl\,{\mathscr N}^{(j)}+{\mathscr M}^{(j)}\big \}_{ti_3i_4\dots i_{j+2}} ,\ \ \ j=1,2,\dots ,
						\end {equation}
					where $\epsilon_{ijk}$ is 
					equal to $1$ if $ijk\in\{ 123,231,312\}$, to $-1$ if $ijk\in\{ 132,213,321\},$ and to zero otherwise.
				\end {Rem}
			
			Substituting (\ref{asymptoticv})	into the homogenised equation (\ref{hom1}),
					we find that the coefficient functions $\bsl{v}_k$
					satisfy the following sequence of recurrence relations:
						\begin {equation}\label {v1}
							 \curl
							 \bigl\{\hat {h}^{(2)}\curl
							 \bsl {v}_0\bigr\}=\bsl {f},\ \ \ \ \ \ \ \ \ 
							\sum_{\substack {j+k=l,\ 
							j, k\in \mathbb {N}\cup\{0\}}}\curl
							\bigl\{\hat {h}^{(j+2)}\del
							^{j}\curl\,
							\bsl {v}_k\bigr\}=\bft {0},\ \ \ l=1,2,...,
						\end {equation}
						\vskip-0.7cm
						\begin {equation}\label {conv5}
							\dive\,\bsl {v}_l=0,\ \ \ \ \ l=0,1,\dots.\ \ \ \ \ \ \ \ \ \ \ \ \ \ \ \ \ \ \ \ \ \ \ \ \ \ \ \ \ \ \ \ \ \ \ \ \ \ \ \ \ \ \ \ \ \ \ \ \ \ \ \ \ \ \ \ \ \ \ \ \ \ \ \ \ \ \  \  \ \ \ \ \ \ \ \ \ \ \ 
						\end {equation}
						\begin {Prop}
							The matrix $\hat {h}^{(2)}$ is symmetric and positive definite. Further, for any given right-hand side $\bsl {f}\in\bigl[C^\infty_{\rm per}(\torus)\bigr]^3\cap {\mathcal X}(\torus)$, there exists a unique solution sequence $\bsl{v}_l,$ $l=0,1,2,...,$ for  (\ref {v1})--(\ref {conv5}) such that all elements of the sequence have zero mean over $\torus.$
					
							\end {Prop}
						\begin {proof}
						Using (\ref{star})
						and the formula (\ref{hom_eq}) for the homogenised matrix $\hat {h}^{(2)},$  
							we write
								\begin{equation}
								\hat {h}^{(2)}=\Bigl\langle A\bigl(\curl\ {\mathscr N}^{(1)}+I\bigr)+\bigl( A\,\curl\ {\mathscr N}^{(1)}+A\bigr) \curl\ {\mathscr N}^{(1)}\Bigr\rangle
								=\Bigl\langle A \bigl( \curl\ {\mathscr N}^{(1)}+I\bigr)\bigl( \curl\ {\mathscr N}^{(1)}+I\bigr)\Bigr\rangle .
								\label{h2_manip}
								\end{equation}
							In view of the fact that the matrix $A$ is symmetric and positive definite, the expression in the right-hand side of (\ref{h2_manip})  
							is also symmetric and positive definite.
							Applying the Lax-Milgram lemma to each equation in (\ref {v1}), in conjunction with the corresponding equation in 
							(\ref {conv5}), yields the second claim of the lemma.
						\end {proof}

			In Section \ref {sec115} we provide a justification of the above formal argument, by showing that the remainder of the truncation of the double series
			 (\ref{A2})--(\ref{asymptoticv}) is close to the solution of the original problem in a certain sense.
			Before doing so, we briefly discuss the recurrence relations that determine the tensors ${\mathscr K}^{(j)}.$ These relations arise by considering the divergence of 
			(\ref {A2}).

		\subsection{Recurrence relations for ${\mathscr K}^{(j)}, j=2,3,...$} 
		\label {sec114}
			Taking the formal divergence of the asymptotic expansion (\ref {A2}) and using the condition ${\rm div}\,{\bsl u}^\varepsilon=0$ yields
				\begin {multline*}
					\sum_{j=1}^\infty \ep^j\Big \{ 
					\bigl({\mathscr K}^{(j)}(\bsl {y})\del^j_{\bsl x}\bsl {v}(\bsl {x},\ep)\bigr)_{,x_iy_i}+
					\bigl({\mathscr K}^{(j-1)}(\bsl {y})\del_{\bsl x}^{j-1}\bsl {v}(\bsl {x},\ep)\bigr)_{,x_ix_i}+\dive_{\bsl x}\bigl({\mathscr N}^{(j)}(\bsl {y})\del_{\bsl x}^{j-1}\curl_{\bsl x}\bsl {v}(\bsl {x},\ep)\bigr)\Big \}\\[-0.7em]
					+\sum_{j=1}^\infty \ep^{j-1}\Big \{ 
					\bigl({\mathscr K}^{(j)}(\bsl {y})\del^j_{\bsl x}\bsl {v}(\bsl {x},\ep)\bigr)_{,y_iy_i}+
					\bigl({\mathscr K}^{(j-1)}(\bsl {y})\del_{\bsl x}^{j-1}\bsl {v}(\bsl {x},\ep)\bigr)_{,x_iy_i}\Big \} =0.
				\end {multline*}
			Comparing the terms with equal powers of $\varepsilon$ yields the system of recurrence relations
				\begin {equation}
						\bigl({\mathscr K}^{(1)}\del_{\bsl x}\bsl {v}\bigr)_{,y_iy_i}
						=0,\label{K1_eq}
						\end{equation}
						\begin{equation} 
					  \bigl({\mathscr K}^{(l+1)}\del^{l+1}_{\bsl x}\bsl {v}\bigr)_{,y_iy_i}+
					  2\bigl({\mathscr K}^{(l)}\del^l_{\bsl x}\bsl {v}\bigr)_{,x_iy_i}+
					  \bigl({\mathscr K}^{(l-1)}\del_{\bsl x}^{l-1}\bsl {v}\bigr)_{,x_ix_i} 
						+\dive_{\bsl x}\bigl({\mathscr N}^{(l)}\del_{\bsl x}^{l-1}\curl_{\bsl x}\bsl {v}\bigr)=0,\  l=1,2,...\ .
						\label {divergence30}
				\end{equation}
			Note that the $Q$-periodic solution ${\mathscr K}^{(1)}$ to (\ref{K1_eq})
				with zero average over $Q$ is identically zero.
			Hence, the first non-trivial tensor in the sequence is the third-order tensor ${\mathscr K}^{(2)},$ which satisfies 
				$$
				\bigl({\mathscr K}^{(2)}\del^2_{\bsl x}\bsl {v}\bigr)_{,y_iy_i}=-\dive_{\bsl x}\bigl({\mathscr N}^{(1)}\curl_{\bsl x}\bsl {v}\bigr),$$
			for an arbitrary vector $\bsl {v}$. 
			Substituting the expansion (\ref{asymptoticv}) 
			into the system (\ref {divergence30}), we obtain
				\begin {multline}\label {asym2}
					\sum_{\substack {j+k=l \\ j\in \mathbb {N},\ k\in \mathbb {N}\cup\{0\}}}\Big \{ 
					\bigl({\mathscr K}^{(j+1)}\del^{j+1}_{\bsl x}\bsl {v}_k\bigr)_{,y_iy_i}+
					2\bigl({\mathscr K}^{(j)}\del^j_{\bsl x}\bsl {v}_k\bigr)_{,x_iy_i}+
					\bigl({\mathscr K}^{(j-1)}\del_{\bsl x}^{j-1}\bsl {v}_k\bigr)_{,x_ix_i}\\[-1.3em]
					+\dive_{\bsl x}\bigl({\mathscr N}^{(j)}\del_{\bsl x}^{j-1}\curl_{\bsl x}\bsl{v}_k\bigr)\Big \} =0,\ \ \ l=1,2,...\ .
				\end {multline}
			We use the set of equations (\ref {asym2}) to establish a bound on the divergence of the solution $\uep $ and, subsequently, 
			on the 
			remainder of the asymptotic series (\ref{main_series}), see Section \ref{sec115}.

		\subsection {Remainder estimates}\label {sec115}
			The justification of the above formal procedure is given by the following theorem.
				\begin {Thrm}\label {thrm1}
					For all 
					$K\in{\mathbb N}$ 
					consider the remainder 
						${\bsl R}^{(K)}({\bsl x},\ep):={\bsl u}^{\ep}(\bsl {x})-{\bsl u}^{(K)}(\bsl {
						x},\ep),$ 
					where 
						\begin{multline}
							\bsl {u}^{(K)}(\bsl {x},\ep):=\bsl {v}^{(K)}(\bsl {x},\ep)+\sum_{j=1}^K\ep^j\Big \{ \del_\bsl{y}\bigl({\mathscr K}^{(j)}(\bsl{y})\del_{\bsl x}^j\bsl {v}^{(K)}(\bsl {x},\ep)\bigr)+ \del_{\bsl x}\bigl({\mathscr K}^{(j-1)}(\bsl{y})\del_{\bsl x}^{j-1}\bsl {v}^{(K)}(\bsl {x},\ep)\bigr)\\ +{\mathscr N}^{(j)}(\bsl{y})\del_{\bsl x}^{j-1}\bigl(\curl_{\bsl x}\bsl {v}^{(K)}(\bsl {x},\ep)\bigr)\Big\}\Bigr\vert_{\bsl{y}=\bsl{x}/\varepsilon},\ \ \ \ \ \ \ \ \ \ \ \ \ \ \ \ \  \ \ \ \ \ \ \quad\quad\quad\quad\quad\label{trunc1} 
\end{multline}						
							\vskip-1.0cm
							\begin{equation}
							{\bsl v}^{(K)}({\bsl x},\ep):=
							\sum^K_{k=0}\ep^k{\bsl v}_k({\bsl x}).
							\quad\quad\quad\quad\quad\quad\quad\quad\quad\quad\quad\quad\quad\quad\quad\quad\quad\quad\quad\quad\quad\quad\quad\quad\quad\quad\quad\quad\quad\quad\label{VK}
						\end {equation}
					Then the estimates
							\begin{align}
							\ & {\rm (i)}\  \,\bigl\| \curl \,{\bsl R}^{(K)}\bigr\|_{L^2(\mathbb {T})}\leq C^{(K)}_1\ep^{K-1},\ \ \ \quad\quad\quad\quad\quad\quad\quad\quad\quad\quad\quad\quad\nonumber\\[0.4em]
							\ & {\rm (ii)}\ \bigl\| \dive\,{\bsl R}^{(K)}\bigr\|_{H^{-1}(\mathbb {T})}\leq C^{(K)}_2\ep^{K},\ \ \ \nonumber\\[0.4em]
							\ & {\rm (iii)}\ \ \forall\,M\ \ \bigl|\bigl\langle {\bsl R}^{(K)}\bigr\rangle_{\mathbb {T}}\bigr|\leq \widetilde {C}^{(K)}_{M}\ep^M,\nonumber
				\end{align}
			hold, where the constants $C^{(K)}_1, C^{(K)}_2, \widetilde {C}^{(K)}_{M}$ are independent of $\varepsilon$ but may depend on the function ${\bsl f}$.
				\end {Thrm}
				\begin {proof}

	(i) We evaluate the operator on the left-hand side of (\ref{firsteq}) on the function $u^{(K)}:$  
						\begin{gather}
							\textbf{curl}\big\{\hat{A}^\varepsilon\textbf {curl}\,\bsl {u}^{(K)}\big\}=
						\sum_{j=0}^{K-2}\ep^j\curl_{\bsl x}\bigl\{\hat {h}^{(j+2)}\curl_{\bsl x}\bsl {v}^{(K)}\bigr\}+\ep^{K-1}\bsl{\Theta}_1\bigl(\bsl {v}^{(K)};\ep,K\bigr),\label{uk451}\\
							\bsl{\Theta}_1(\bullet
							;\ep,K):=\ep\,\Bigl(\curl_{\bsl x}\bigl\{A\,\curl_{\bsl x}\bigl\{\del_\bsl{y}\bigl ({\mathscr K}^{(K)}\del_{\bsl x}^K
							 \bullet\bigr)\bigr\}+{\mathscr N}^{(K)}\del_{\bsl x}^{K-1}\curl_{\bsl x} \bullet
							\bigr\}\Bigr)\ \ \ \ \ \ \ \ \ \ \ \ \ \ \ \ \ \ \nonumber
							\\[0.7em]\ \ \ \ \ \ \ \ \ \ \ \ \ \ \ +\curl_\bsl{y}\bigl\{A\,\curl_{\bsl x}\del_\bsl{y}\bigl( {\mathscr K}^{(K)}\del_{\bsl x}^K
							 \bullet\bigr)\bigr\}+
							\bigl\{\curl_\bsl{y}\bigl(A{\mathscr M}^{(K+1)}\bigr)+{\mathscr L}^{(K+1)}\bigr\}\del_{\bsl x}^K\curl_{\bsl x}\bullet\Bigr\vert_{\bsl{y}=\bsl{x}/\varepsilon}.
							\nonumber
						\end {gather}
					Substituting the expression for $\bsl{v}^{(K)}$ from (\ref {VK}) into (\ref {uk451}), we obtain
						\begin {multline*}
							\textbf{curl}\big\{\hat{A}^\varepsilon\textbf {curl}\,\bsl {u}^{(K)}\big\}=
							\sum_{j=0}^{K-2}\sum_{k=0}^K\ep^{j+k}\curl
							\bigl\{\hat {h}^{(j+2)}\del
							^j\curl\,
							\bsl {v}_k\bigr\}
							\\[-0.1em]
							+\ep^{K-1}\sum_{k=0}^K\ep^k\bsl{\Theta}_1(\bsl {v}_k;\ep,K)=
						\bsl {f}+\ep^{K-1}\bsl{\Theta}_2(\bsl {x},\ep,K),\ \ \ \ \ \ \ \ \ \ \ \ \ \ \ \ \ \ \ \ 
						\end {multline*}
						$$\bsl{\Theta}_2(\bsl {x},\ep,K):=
						\sum_{j+k=K-1}^{2K-2}\ep^{j+k-K+1}\curl
						\bigl\{\hat {h}^{(j+2)}\del
						^j\curl\,
						\bsl {v}_k\bigr\}+\sum_{k=0}^K\ep^k\bsl{\Theta}_1(\bsl {v}_k;\ep,K),\ \ \ \ \ \ \ \ \ \ \ \ \ \ \ \ \ \ \ \ \ \ \ \ \ \ \ \ \ \ 
						$$
					where 
					we use (\ref {v1}).
					It can be shown that 
					\[
					\bigl|\bsl{\Theta}_2(\bsl {x},\ep,K)\bigr|\leq\bar{C}^{(K)},\ \ \ \ \ \ \bsl{x}\in\torus,
					\] 
					with an $\varepsilon$-independent constant $\bar{C}^{(K)}>0.$ 
				Recalling (\ref{firsteq}), we obtain
						\begin {equation}\label {remain1}
							\textbf{curl}\big\{\hat{A}^\varepsilon(\bsl{x})\textbf {curl}\,\bsl {R}^{(K)}(\bsl{x},\varepsilon)\big\}=-\ep^{K-1}\bsl{\Theta}_2(\bsl {x},\ep,K).
						\end {equation}								
					Further, taking the scalar product of both sides of equation (\ref {remain1}) with $\bsl {R}^{(K)},$ integrating over $\mathbb {T},$ and 
					using the Cauchy-Schwarz inequality on the right-hand side yields
						$$\int_{\mathbb {T}}\,\hat{A}^\varepsilon
						\textbf {curl}\,\bsl {R}^{(K)}\cdot \textbf {curl}\,\bsl {R}^{(K)}
						=
					-\ep^{K-1}\int_{\mathbb {T}}\bsl{\Theta}_2\cdot \bsl {R}^{(K)}
						\leq \ep^{K-1}\| \bsl{\Theta}_2\|_{[L^2(\mathbb {T})]^3}\bigl\| \bsl {R}^{(K)}\bigr\|_{[L^2(\mathbb {T})]^3}.$$
					Since the matrix $\hat{A}^\varepsilon$ is positive definite, it follows that
						$$\nu\bigl\|\curl\,\bsl {R}^{(K)}\bigr\|_{[L^2(\mathbb {T})]^3}^2\leq \ep^{K-1}\| \bsl{\Theta}_2\|_{[L^2(\mathbb {T})]^3}\bigl\| \bsl {R}^{(K)}\bigr\|_{[L^2(\mathbb {T})]^3},$$
					and using the Poincar\'{e}-type inequality from Appendix A yields
						$$\nu\bigl\|\curl\,\bsl {R}^{(K)}\bigr\|_{[L^2(\mathbb {T})]^3}^2\leq \ep^{K-1}\vert\torus\vert\| \bsl{\Theta}_2\|_{[L^2(\mathbb {T})]^3}\Big (\bigl\| \curl\,\bsl {R}^{(K)}\bigr\|_{[L^2(\mathbb {T})]^3}+\bigl\| \dive\,\bsl {R}^{(K)}\bigr\|_{L^2(\mathbb {T})}\Big).$$
					 It will be shown in the proof of (ii) that $\bigl\|\dive\,\bsl {R}^{(K)}\bigr\|_{L^2(\mathbb {T})}=O\bigl(\ep^{K}\bigr)$. Hence 
						$$\bigl\|\curl\,\bsl {R}^{(K)}\bigr\|_{[L^2(\mathbb {T})]^3}^2\leq\vert\torus\vert\| \bsl{\Theta}_2\|_{[L^2(\mathbb {T})]^3}\nu^{-1}\ep^{K-1}\bigl\|\curl\,\bsl {R}^{(K)}\bigr\|_{[L^2(\mathbb {T})]^3}+O\bigl(\ep^{2K-1}\bigr),$$
						which implies the claim.

				(ii)	Note that $\dive\,\bsl {u}^\ep=0$ and hence
						$\dive\,\bsl {R}^{(K)}=-\dive \,\bsl {u}^{(K)}.$
					 Denoting by 
					  $\bsl {U}_{l}$ the coefficient in front of $\varepsilon^l$ in the combined sum (\ref {trunc1})--(\ref{VK}):
						\begin{equation}
						\bsl {U}_{0}:=\bsl {v}_0,\ \ 
						\bsl {U}_{l}:=\sum_{j=1}^K\sum_{k=0}^{\min\{l-j, K\}}\Bigl\{\del_\bsl{y}\bigl({\mathscr K}^{(j)}\del_{\bsl x}^j\bsl {v}_k\bigr)+ \del_{\bsl x}\bigl({\mathscr K}^{(j-1)}\del_{\bsl x}^{j-1}\bsl {v}_k\bigr)+{\mathscr N}^{(j)}\del_{\bsl x}^{j-1}\bft {curl}_{\bsl x}\bsl {v}_k\Bigr\},\ \ \ l=1,2,...,
						\label{ujkeq}
						\end{equation}
					we write
						\begin {equation*}
							-\dive\,\bsl {R}^{(K)}(\bsl{x}, \varepsilon)=\sum_{l=1}^{2K}\bigl\{\ep^{l}\dive_{\bsl x}\bsl {U}_{l}+\ep^{l-1}\dive_\bsl{y}\bsl {U}_{l}\bigr\}\bigr\vert_{\bsl{y}=\bsl{x}/\varepsilon}.
	\end {equation*}
					Taking into account the recurrence relations (\ref{asym2}), we obtain
						$$-\dive\,\bsl {R}^{(K)}(\bsl{x}, \varepsilon)
						=\ep^{K} \Theta_3(\bsl {x},\ep, K),\ \ \ \ \ \ \ \ 
					 \Theta_3(\bsl {x},\ep, K):=
						\biggl\{\,\sum_{l=K
						}^{2K}\ep^{
						l-K}\dive_{\bsl x}
						\bsl{U}_l
						+\sum_{
			          		l=K+1
						}^{2K}\ep^{
						l-K-1}\dive_\bsl{y}\bsl{U}_l	\biggr\}\biggr\vert_{\bsl{y}=\bsl{x}/\varepsilon}								
						.$$
					Note first, see (\ref{ujkeq}), that $\bsl{U}_l$ is a finite sum of 
					terms of the form
						$U(\bsl {y})V(\bsl {x}),$
					for some tensors $U(\bsl {y})$ with elements in $H^1(Q)$ and $V(\bsl {x})$ with elements in $C^\infty_{\rm per}(\torus )$. Further, since
						$$\dive_\bsl{y} \bsl {U}_{l}\bigr\vert_{\bsl{y}=\bsl{x}/\varepsilon}=\ep \dive\,\bigl(\bsl {U}_{l}\bigr\vert_{\bsl{y}=\bsl{x}/\varepsilon}\bigr)-\ep \dive_{\bsl x}\bsl {U}_{l}\bigr\vert_{\bsl{y}=\bsl{x}/\varepsilon},$$
						the two-scale form of $ \Theta_3$ is a finite sum of terms of the form $\widetilde{U}(\bsl {y})\widetilde{V}(\bsl {x}),$ for some tensors $\widetilde{U}(\bsl {y})$ with elements in $L^2(Q)$ and $\widetilde{V}(\bsl {x})$ with elements in $C^\infty (\torus ).$
						Finally, we use the following statement, which is a version of the theorem in \cite[Appendix C]{bib13}.
						\begin {Lem}
				               \label{RL_type}
						Let $M(\bsl {x}/\ep)$ be a periodic tensor of order $j$ whose components have zero average on $Q$ and let $g(\bsl {x})$ be a smooth, periodic tensor of order $j-1$. Then there exist positive constants $C_r,$ $r\in{\mathbb N},$ such that
							\begin {equation}\label {monster}
								\left | \int_\torus M(\bsl {x}/\ep)g(\bsl {x})\dx \right | \leq C_r\ep^r, \ \ \ \ \ \ \forall r\in \mathbb {N}.
							\end {equation}
					\end {Lem}
						It follows from  the above lemma that the $L^2(\torus)$-norm of $ \Theta_3$ is bounded by an $\ep$-independent 
						constant $C_2^{(K)},$ and hence 
						$$\bigl\| \dive\,\bsl {R}^{(K)}\bigr\|_{L^2(\torus )}\leq \ep^K\|  \Theta_3\|_{L^2(\torus )}\leq C_2^{(K)}\ep^K.$$


				(iii)	
					Integrating (\ref{trunc1}) over $\torus$ and using $Q$-periodicity of ${\mathscr K}^{(j)}$ and $\torus$-periodicity of $\bsl {v}$ yields
						$$\int_\torus \bsl {u}^{(K)}(\bsl {x})\,\dx =\sum_{j=0}^K\ep^j
						\int_\torus{\mathscr N}^{(j)}(\bsl {x}/\ep )\del
						^{j-1}\bft {curl}\,
						\bsl {v}^{(K)}(\bsl {x})\,\dx .$$
					The result follows by applying once again Lemma \ref{RL_type}.
			\end {proof}
			\begin {Corol}
			\label{corol1}
				The estimates
					\begin{equation}
					\bigl\|\bsl {R}^{(K)}\bigr\|_{[L^2(\torus )]^3}\leq C^{(K)}_3\ep^{K-1},\ \ \ \ \ 
					\bigl\| \bsl {R}^{(K)}\bigr\|_{H^1_\curl(\torus )}\leq C^{(K)}_4\ep^{K-1},\ \ \ \ \ \ \ C^{(K)}_3, C^{(K)}_4>0,
					\label{RK_L2}
					\end{equation}
			hold, where  $C^{(K)}_3$ and $C^{(K)}_4$ are $\varepsilon$-independent but may depend on the function $\bsl {f}$.
			\end {Corol}
			\begin {proof}
				The proof is an immediate consequence of the theorem and the inequality (\ref{Maxwell_inequality}).
			\end {proof}


		\begin{Rem} The above method of truncating (\ref {hom1}) may lead to non-elliptic higher-order problems.
		 An alternative approach, which is free of this limitation, is discussed in Section \ref {sec12}. 
		\end{Rem}

		\subsection {Example: laminate with two layers per period}\label {sec116}
			Suppose  that $A (\bsl {y})$ is given by
				$$A (\bsl {y} )=\alpha (y_2) I,\ \ \ \ \alpha (y_2) = \left \{ \begin {array}{ll} \alpha_1, & 0\leq y_2\leq l_1 , \\ \alpha_2, & l_1< y_2\leq 1, \end {array} \right.$$
			where $0<l_1<1,$ and $\alpha_1,$ $\alpha_2$ are positive constants. 
			In what follows, we determine the leading order terms in the asymptotic expansion (\ref {hom1}) 
				$$
				\textbf{curl}\big\{\hat{A}^\varepsilon\textbf {curl}\,\bsl {u}^\varepsilon\big\}
				=\curl_{\bsl x}\bigl\{\hat {h}^{(2)}\curl_{\bsl x}\bsl {v}\bigr\}+\ep \curl_{\bsl x}\bigl\{\hat {h}^{(3)}\del_{\bsl x}\curl_{\bsl x}\bsl {v}\bigr\}
				+\ep^2 \curl_{\bsl x}\bigl\{\hat {h}^{(4)}\del_{\bsl x}^2\curl_{\bsl x}\bsl {v}\bigr\}+O(\ep^3),$$
				for the case above.
			For the term of order $O(1)$, the related tensors 
			are given by
				$${\mathscr N}^{(1)}=\begin {pmatrix}
				0 & 0 & -N \\
				0 & 0 & 0 \\
				N & 0 & 0\end {pmatrix},\ \ \ 
				\hat {h}^{(2)}=\begin {pmatrix}
				\langle \alpha^{-1}\rangle^{-1} & 0 & 0 \\
				0 & \langle \alpha \rangle & 0 \\
				0 & 0 & \langle \alpha^{-1}\rangle^{-1}\end {pmatrix},$$
			where $N=N(y_2)$ satisfies the differential equation $-(\alpha N')'=\alpha '$ see Remark \ref{Sec3_remark}, and
				$$\langle \alpha^{-1}\rangle^{-1}=\bigl\{\alpha_1^{-1}l_1+\alpha_2^{-1}(1-l_1)\bigr\}^{-1},\ \ \ \ \langle \alpha\rangle=\bigl(\alpha_1l_1+\alpha_2(1-l_1)\bigr).$$
			For the term of order $O(\ep)$, we have
				$${\mathscr N}^{(2)}_{ijk} = \left \{ \begin {array}{ll} M, &  ijk=\{ 123\} , \\ -M, & ijk=\{ 321\} , \\ L, & ijk=\{ 132\} , \\ -L, & ijk=\{ 312\} , \\ 0, & \text {otherwise}, \end {array} \right.\ \ \ \ \ 
				\hat {h}^{(3)}_{ijk}= \left \{ \begin {array}{ll} a, &  ijk=\{ 112,\ 332\} , \\ -b, & ijk=\{ 211,\ 233\} , \\ 0, & \text {otherwise}, \end {array} \right .$$
			where 
					\begin{equation}
					(\alpha M')'=(\alpha N)',\ \ \ \ \  
					(\alpha L')'=\langle \alpha \rangle -\alpha,\ \ \ \ \ 
					a=-\langle \alpha L'\rangle,\ \ \ \ b=\langle \alpha N\rangle .
					\label{MLab}
					\end{equation}
			Finally, the relevant tensors for the term of order $O(\ep^2)$ are given by
				$${\mathscr N}^{(3)}_{ijkl} = \left \{ \begin {array}{ll} P, &  ijkl=\{ 1232\} , \\ -P, & ijkl=\{3212\} , \\ Q, & ijkl=\{ 1223\} , \\ -Q, & ijkl=\{ 3221\} , \\ R, & ijkl=\{ 1311,1333\} , \\ -R, & ijkl=\{ 3111,3133\} , \\ 0, & \text {otherwise}, \end {array} \right.\ \ \ \ \ 
				\hat {h}^{(4)}_{ijkl}= \left \{ \begin {array}{ll} c, &  ijkl=\{ 1212,\ 3232\} , \\ d, & ijkl=\{2121,\ 2323\} , \\ e, & ijkl=\{2112,\ 2332\} , \\ f, & ijkl=\{1111,1133,3311,3333\} , \\ 0, & \text {otherwise}, \end {array} \right .$$
			where 
					\[
					-(\alpha P')'=a+\alpha L'+(\alpha L)',\ \ \ \ 
					-(\alpha Q')'=(\alpha M)',\ \ \ \  
					-(\alpha R')'=b-\alpha N, 
					\]
					\[
					c=-\langle \alpha P'+\alpha L\rangle ,\ \ \ \ d=\langle \alpha Q'\rangle ,\ \ \ \ \ e=\langle \alpha L \rangle,\ \ \ \ f=-\langle \alpha R'\rangle.
					\]
			It is shown by direct calculation that
				\begin {align*}
					a=b=0,\ \ \ \ \ \ \ \ \ \ c=-d=-12^{-1}l_1^2l_2^2\bigl(&\alpha_1^{-1}l_1+\alpha_2^{-1}l_2\bigr)^{-1}(\beta_1-1)(\beta_2-1),\\[0.6em]
					e=  12^{-1}l_1^2l_2^2\bigl(\alpha_1^{-1}l_1+\alpha_2^{-1}l_2\bigr)(\alpha_1-\alpha_2)^2, & \ \ \ \ \ \ 
					f=  12^{-1}l_1^2l_2^2\bigl(\alpha_1^{-1}l_1+\alpha_2^{-1}l_2\bigr)^{-2}(1-\beta_1)(1-\beta_2)(\alpha_2^{-1}l_1+\alpha_1^{-1}l_2),
				\end {align*}
			where $l_1+l_2=1$, and $\beta_1=\alpha_2^{-1}\alpha_1=\beta_2^{-1}$. 
			 The homogenised equation takes the form
				\begin {multline*}
					\begin {pmatrix}
						\avalphin (v_{2,12}-v_{1,22})-\avalph (v_{1,33}-v_{3,13}) \\[0.3em]
						\avalphin (v_{3,23}-v_{2,33}-v_{2,11}+v_{1,12}) \\[0.3em]
						\avalph (v_{1,13}-v_{3,11})-\avalphin (v_{3,22}-v_{2,23}) 
					\end {pmatrix} \\[0.5em]
					+\ep^2\begin {pmatrix}
					f (v_{3,1223}-v_{1,2233})-e (v_{1,1133}-v_{3,1113}+v_{1,3333}-v_{3,1333}) \\[0.2em]
					0 \\[0.1em]
					-f (v_{3,1122}-v_{1,1223})+e (v_{1,1113}-v_{3,1111}+v_{1,1333}-v_{3,1133})
					\end {pmatrix}+O(\ep^3)={\bsl f}$$
				\end {multline*}
				\begin {Rem}
					It was shown in \cite {bib13} that in the case of a scalar equation, all terms with odd powers of $\ep$ are absent from the corresponding infinite-order homogenised equation. The above two-layered case  provides an example in which there is a non-trivial term of order $O(\ep^3)$ in the homogenisation procedure for the Maxwell system. It is calculated that
					$$\scalemath{0.85}{{\mathscr N}^{(4)}_{ijklm} = \left \{ \begin {array}{ll} N_1, &  ijklm=\{ 12232\} , \\ -N_1, & ijklm=\{32212\} , \\ N_2, & ijklm=\{ 13121,13323\} , \\ -N_2, & ijklm=\{ 31121,31323\} , \\ N_3, & ijklm=\{ 13332,13112\} , \\ -N_3, & ijklm=\{ 31332,31112\} 					, \\ N_4, & ijklm=\{ 12311,12333\} , \\ -N_4, & ijklm=\{ 32111,32133\} , \\ N_5, & ijklm=\{ 12223\} , \\ -N_5, & ijklm=\{ 32221\} , \\ N_6, & ijklm=\{ 23212\} , \\ -N_6, & ijklm=\{ 21232\} , \\ N_7, & ijklm=\{ 23111,23133\} , \\ -N_7, & ijklm=\{ 21311,21333\} ,\\ 						0, & \text {otherwise}, \end {array} \right.\ \ \ \ \ 
					\ \ \hat {h}^{(5)}_{ijklm}= \left \{ \begin {array}{ll} h_1, &  ijklm=\{ 12212,32232\} , \\ h_2, & ijklm=\begin {Bmatrix} 11121,\ 11323 \\ 33121,\ 33323\end {Bmatrix} , \\ h_3, & ijklm=\begin {Bmatrix}11112,\ 11332\\ 33112,\ 33332\end {Bmatrix} , \\ h_4, & 							ijklm=\{21212,23232\} , \\ h_5, & ijklm=\{21221,23223\} , \\ h_6, & ijklm=\begin {Bmatrix} 21111,\ 21133\\ 23311,\ 23333\end {Bmatrix} ,  \\ h_7, & ijklm=\begin {Bmatrix}12111,\ 12133 \\ 32311,\ 32333\end {Bmatrix} , \\ 0, & \text {otherwise}, \end {array} 					\right .}$$
				so that the term in question is given by
					\begin {multline*}
						\begin {pmatrix}
							(h_1-h_2+h_5-h_7) (v_{1,22233}-v_{3,12223})+(h_3-h_4+h_6) (v_{1,11233}-v_{3,11123}+v_{1,23333}-v_{3,12333}) \\[0.2em]
							0 \\[0.1em]
							-(h_1-h_2+h_5-h_7) (v_{1,12223}-v_{3,11222})-(h_3-h_4+h_6) (v_{1,12333}-v_{3,11233}+v_{1,11123}-v_{3,11112}).
						\end {pmatrix}
							%
					\end {multline*}
				The functions $N_i,\ i=1,\dots ,7$, are functions of $y_2$ only, which satisfy 
					$$-(\alpha N_1')'=(\alpha P)'+\alpha P'+\alpha L+c,\ \ \ -(\alpha N_2')'=\alpha M+d,\ \ \ -(\alpha N_3')'=\alpha L-e,$$
					$$-(\alpha N_4')'=(\alpha R)'+\alpha R'+f,\ \ \ -(\alpha N_5')'=(\alpha Q)',\ \ \ -(\alpha N_6')'=\alpha P'+\alpha L+c,\ \ \ -(\alpha N_7')'=\alpha R'+f,$$
				and the constants $h_i,\ i=1,\dots ,7,$ are given by
					$$h_1=-\langle \alpha N_1'+\alpha P\rangle , \ \ \ h_2=-\langle \alpha N_2'\rangle , \ \ \ h_3=-\langle \alpha N_3'\rangle ,$$
					$$h_4=\langle \alpha P\rangle ,\ \ \ h_5=\langle \alpha Q\rangle ,\ \ \ h_6=\langle \alpha R\rangle ,\ \ \ h_7=-\langle \alpha N_4'+\alpha R\rangle.$$
				It is shown directly that $h_1=-h_5$, $h_2=-h_7$ and $h_3=-h_6,$ hence the above order $O(\ep^3)$ term takes the form
					$$\begin {pmatrix}
							-h_4(v_{1,11233}-v_{3,11123}+v_{1,23333}-v_{3,12333}) \\[0.2em]
							0 \\[0.1em]
							h_4(v_{1,12333}-v_{3,11233}+v_{1,11123}-v_{3,11112})
					\end {pmatrix}.$$
				\end {Rem}

	\section {Variational approach}\label {sec12}

				Consider the minimisation problem
				\begin {equation}\label {var1}
					\min_{\bsl{u}
					}\int_{\torus}\Big ( \frac {1}{2}\hat{A}^\varepsilon\,
					\curl\,\bsl {u} \cdot \curl \,\bsl {u} -\bsl {f} \cdot \bsl {u} \Big ) 
					=:I(\ep ,\bsl {f}),
				\end {equation}
			where the minimum is taken over  $\torus$-periodic functions with zero average and zero divergence. 
			Clearly (\ref{firsteq}) is the Euler-Lagrange equation for the functional (\ref {var1}). Therefore, one has 
				\begin {equation}\label {var2}
					I(\ep ,\bsl {f})=-\frac {1}{2}\int_{\torus}\hat{A}^\varepsilon\,
					\curl \,\uep \cdot \curl \,\uep 
					=-\frac {1}{2}\int_\torus \bsl {f} \cdot \uep.
				\end {equation}


		\subsection {Variational asymptotics}\label {sec122}
			Similarly to (\ref {var1}), the solution 
			to the homogenised equation (\ref{hom_eq}) 
	is the solution to the minimisation problem
				\begin {equation}\label {var3}
					\min_{\bsl{v}
					}\int_\torus \Big ( \frac {1}{2}\hat {h}^{(2)}\curl\,
					\bsl {v}
					\cdot \curl\,
					\bsl {v}
					-\bsl {f}
					\cdot \bsl {v}
					\Big)
					=:I_0(\bsl {f}) ,
				\end {equation}
			over all divergence free, periodic functions on $\torus $. 
			It is well-known, see {\it e.g.} \cite{bib6}, that the functional (\ref {var1}) converges to the homogenised functional (\ref {var3}) as $\ep \rightarrow 0$, {\it i.e.} for any function $\bsl {f}$ one has
				$I(\ep ,\bsl {f})\rightarrow I_0(\bsl {f})$ as $\ep \rightarrow 0.$ 
			The following result generalises this fact to all finite orders in $\varepsilon.$ 
				\begin {Prop}\label {prop001}
  				For any function $\bsl {f}\in\bigl[C^\infty_{\rm per}(\torus)\bigr]^3$ and any positive integer K, there exists a constant $\widehat{C}_K$ such that
						\[
						\biggl\vert I(\ep ,\bsl {f})+\frac {1}{2}\int_\torus \bsl {f}(\bsl {x})\cdot \bsl {v}^{(K)}(\bsl {x},\ep )\ \mathrm {d}\bsl {x}\biggr\vert\leq \widehat{C}_K\ep^K.
						\]
				\end {Prop}
				\begin {proof}
                        Using Theorem \ref{thrm1},  we write  
					$$I(\ep ,\bsl {f})=-\frac {1}{2}\int_\torus \bsl {f} \cdot \bsl {u}^{(K)}
					 +\int_{\torus}\bsl {R}^{(K)}(\bsl {x},
					\ep)\dx$$
						\begin{equation}
						\label{energy_estimate}
						=-\frac {1}{2}\int_\torus \bsl {f}(\bsl {x})\cdot \bsl {v}^{(K)}(\bsl {x},\ep )\mathrm {d}\bsl {x}+\sum_{\substack {1\leq j\leq K \\ 0\leq k\leq K}}\ep^{j+k}\int_\torus {\mathscr N}^{(j)}(\bsl {x}/\ep )F_{jk}(\bsl {x})\mathrm {d}\bsl {x}+\int_{\torus}\bsl {R}^{(K)}(\bsl {x},
						\ep)\dx,
						\end{equation}
						where 
						\[
						\biggl|\int_{\torus}\bsl {R}^{(K)}(\bsl {x},
						\ep)\dx
				\biggr|\leq \widetilde {C}_K\ep^K,\ \ \ \ \ \ \ \ \widetilde {C}_K>0,
                                                   \]
						and $F_{jk}:=\del_{\bsl x}^{j-1}\curl_{\bsl x}\bsl {v}_k*\bsl {f}$ are infinitely smooth and $\torus$-periodic. 
				As the tensors ${\mathscr N}^{(j)}$ are periodic and have zero mean, the integrals                                                   \begin {equation*}
							\ep^{j+k}\int_\torus {\mathscr N}^{(j)}(\bsl {x}/\ep )F_{jk}(\bsl {x})\mathrm {d}\bsl {x}, \ \ \ \ \ \ \ \ \ 1\leq  j\leq K,\ 0\leq  k\leq K,
						\end {equation*}
			go to zero as $\ep \rightarrow 0$ faster than any power of $\ep,$ see Lemma \ref{RL_type}.
						This observation and (\ref{energy_estimate}) immediately imply the claim.
 															\end {proof}
			The above proposition can be interpreted in the sense that 
				\begin {equation}
					I(\ep ,\bsl {f})\,\stackrel{\varepsilon\to0}{\sim}\,
					\sum_{k=1}^\infty \ep^kI_k(\bsl {f}),\ \ \ \ \ \ \ \ \ \ 
				I_k(\bsl {f}):=-\frac {1}{2}\int_\torus \bsl {f}
				\cdot \bsl {v}_k
				\ \ \ \ \ \ \ \ k=0,1,2,\dots.
				\label {I1}
				\end{equation}

		\subsection {Infinite-order homogenised solution}\label {sec123}
			In this section, the effect of the rapid oscillations in the tensors ${\mathscr N}^{(j)}$ in the asymptotic expansion (\ref {A2}) is removed by 
			following the translation averaging strategy of \cite{bib13}.
			More precisely, for all 
			$ \boldsymbol {\zeta}\in Q,$ denote $A_\zeta (\bsl {y}):=A(\bsl {y}+ \boldsymbol {\zeta}),$ $\bsl {y}\in Q,$ and consider the equation
				\begin {equation}\label {zeta1}
					\curl\bigl\{ \hat{A}_\zeta^\varepsilon(\bsl{x})\,
					\curl \,\bsl {u} (\bsl {x})\bigr\}=\bsl {f}(\bsl {x}),\ \ \ \ \ \bsl{x}\in{\mathbb T},
				\end {equation}
		 under the same assumption for the right-hand side $\bsl {f}$ and subject to the same conditions on the solution 
		 $\bsl {u}$ as in the case of (\ref{firsteq}). 
			For all $ \boldsymbol {\zeta}$, the equation (\ref {zeta1}) admits a unique solution $\uep_\zeta$. 
			 Consider the average of $\uep_\zeta$ with respect to $ \boldsymbol {\zeta}$, {\it i.e.}
				\begin {equation}\label {zeta2}
					\bar {\bsl {u}}^\ep(\bsl {x}) :=\bigl\langle \uep_\zeta(\bsl{x})\bigr\rangle_\zeta=\frac {1}{|Q|}\int_Q\uep_\zeta(\bsl {x})\ \mathrm {d} \boldsymbol {\zeta},\ \ \ \ \ \bsl{x}\in{\mathbb T}.
				\end {equation}
			This averaging with respect to $ \boldsymbol {\zeta}$ is analogous to ``ensemble averaging"  for the family (\ref{zeta1}) when the underlying probability 
			measure is uniform over $Q.$ 
			The next proposition shows that translation averaging eliminates the effects of oscillations due to the tensors ${\mathscr N}^{(j)}$ and results in the validity of the asymptotics (\ref {asymptoticv}) for the average $\bar {\bsl {u}}^\ep.$ 
				\begin {Prop}\label {prop2}
					For a given function $\bsl {f}$, the series (\ref{asymptoticv}) provides asymptotics for the function $\bar {\bsl {u}}^\ep$ in the sense that for any positive integer $K$, there exists a positive constant $C_K$ such that
						\begin {equation*}
							\bigl\Vert\bar {\bsl {u}}^\ep-\bsl {v}^{(K)}\bigr\Vert_{[L^2(\torus)]^3}
							\leq C_K\ep^{K}.
						\end {equation*}
				\end {Prop}
				\begin {proof}
					Note that for all $ \boldsymbol {\zeta}\in Q$,  an expansion similar to (\ref {A2})--(\ref {asymptoticv}) holds 
					for the solution $\uep_\zeta$ of (\ref {zeta1}),
					where ${\mathscr K}^{(j)}$ and ${\mathscr N}^{(j)}$ are replaced by the ``translated'' tensors ${\mathscr K}^{(j)}_\zeta(\cdot):={\mathscr K}^{(j)}(\cdot+ \boldsymbol {\zeta})$ and  ${\mathscr N}^{(j)}_\zeta (\cdot):={\mathscr N}^{(j)}(\cdot+ \boldsymbol {\zeta}),$ respectively. 
					In particular, by analogy with the first estimate in (\ref{RK_L2}), one has  
						$$\bigl\| \uep_\zeta-\bsl {u}^{(K)}_\zeta\bigr\|_{[L^2(\torus)]^3}\leq \bar {C}_K \ep^K,\ \ \ \ \ \ \ \ \ \ \bar {C}_K>0,$$
					where  
					$\bsl {u}^{(K)}_\zeta$ is a truncation of the asymptotic series for $\uep_\zeta$ similar to (\ref {trunc1}), with all $\bsl{y}$-dependent objects replaced by their 
					$\boldsymbol{\zeta}$-translated versions:
						\begin {multline}
						\label{uK_sum}
							\bsl {u}_\zeta^{(K)}(\bsl {x},\ep)=\bsl {v}^{(K)}(\bsl {x},\ep)+\sum_{j=1}^K\ep^j\Big \{ \del_\bsl{y}\bigl({\mathscr K}_\zeta^{(j)}(\bsl {y})\del_{\bsl x}^j\bsl {v}^{(K)}(\bsl {x},\ep)\bigr)\\ 
							+ \del_{\bsl x}\bigl({\mathscr K}_\zeta^{(j-1)}(\bsl {y})\del_{\bsl x}^{j-1}\bsl {v}^{(K)}(\bsl {x},\ep)\bigr) +{\mathscr N}_\zeta^{(j)}(\bsl {y})\del_{\bsl x}^{j-1}\bigl(\bft {curl}_{\bsl x}\bsl {v}^{(K)}(\bsl {x},\ep)\bigr)\Big \} .
						\end {multline}
					Next, we note that
						$$\bar {\bsl {u}}^\ep(\bsl{x})-\bsl {v}^{(K)}(\bsl{x},\varepsilon)=\int_Q \bigl( \uep_\zeta(\bsl {x})-\bsl {u}_\zeta^{(K)}(\bsl {x},\ep)\bigr) \mathrm {d} \boldsymbol {\zeta},$$
					since the integrals of the $\bsl{y}$-dependent tensors under summation in (\ref{uK_sum}) clearly vanish. Therefore, one has
						$$\int_\torus \big | \bar {\bsl {u}}^\ep(\bsl{x})-\bsl {v}^{(K)}(\bsl{x},\varepsilon)\big |^2\mathrm {d}\bsl {x}\leq \int_\torus \biggl( \int_Q \big | \uep_\zeta(\bsl {x})-\bsl {u}_\zeta^{(K)}(\bsl {x},\ep)\big | \mathrm {d} \boldsymbol {\zeta}\biggr)^2\mathrm {d}\bsl {x}\ \ \ \ \ \ \ \ \ \ \ \ \ \ \ \ \ \ \ \ \ \ \ \ \ \ \ \ \ \ \ \ 
						$$
						\[
						\ \ \ \ \ \ \ \ \ \ \ \ \ \ \ \ \ \ \ \ \ \ \ \ \ \ \ \ \ \ \ \ \leq \int_Q \int_\torus \big | \uep_\zeta(\bsl {x})-\bsl {u}_\zeta^{(K)}(\bsl {x},\ep)\big |^2\mathrm {d}\bsl {x} \mathrm {d} \boldsymbol {\zeta}\leq C_K'\ep^{2K},\ \ \ \ \ \ \ \ C_K'>0,
						\]
					as required.
				\end {proof}

		\subsection {Higher-order variational problems} \label {sec124}
			By analogy with (\ref{var1}), we define the ``translated'' energy functional:
				\begin {equation}\label {zeta4}
					I_\zeta (\ep ,\bsl {f}):
					=\min_{\bsl {u}(\bsl {x})}\int_{\torus}\Big ( \frac {1}{2} A_\zeta^\varepsilon\,					
					\curl \ \bsl {u} \cdot \curl \ \bsl {u} -\bsl {f} \cdot \bsl {u} \Big ).
				\end {equation}
			and the ``averaged'' functional 
				\begin {equation*}
					\bar {I}(\ep, \bsl {f}):=\int_Q I_\zeta (\ep ,\bsl {f})\ \mathrm {d} \boldsymbol {\zeta}.
				\end {equation*}
			Considering $ \boldsymbol {\zeta}$-dependent trial fields in (\ref{zeta4}) and changing the order of $ \boldsymbol {\zeta}$-integration and minimisation yields
				\begin {equation}\label {I2}
					\bar {I}(\ep, \bsl {f})
					=\min_{\bsl {u}(\bsl {x}, \boldsymbol {\zeta})}\bar {E}_{\ep }(\bsl {u},\bsl {f}),
				\end {equation}
			where 
				\begin {equation}\label {zeta7}
					\bar {E}_{\ep }(\bsl {u},\bsl {f}):=\int_Q\int_{\mathbb T}\Big ( \frac {1}{2}\hat{A}_\zeta^\varepsilon(\bsl{x})
					\curl \,\bsl {u}(\bsl{x}) \cdot \curl \,\bsl {u} (\bsl{x})-\bsl {f}(\bsl{x})\cdot\bsl {u} (\bsl{x})\Big ) \dx
					\mathrm {d} \boldsymbol {\zeta}.
				\end {equation}
			Clearly, the variational problem (\ref {I2}) has as its minimiser the function $\bsl {u}(\bsl {x}, \boldsymbol {\zeta})=\uep_\zeta(\bsl {x}),$ where $\uep_\zeta(\bsl {x})$ is the solution of (\ref {zeta1}). Further, recall that $\uep_\zeta$ is represented by the series 
				\begin{equation*}
					\bsl {u}^\ep _\zeta(\bsl {x})=\bsl {v}(\bsl {x})+\sum_{j=1}^\infty \ep^j\Big \{ \del_\bsl{y}\bigl({\mathscr K}_\zeta ^{(j)}(\bsl {y})\del_{\bsl x}^j\bsl {v}(\bsl {x})\bigr)+ \del_{\bsl x}\bigl({\mathscr K}_\zeta ^{(j-1)}(\bsl {y})\del_{\bsl x}^{j-1}\bsl {v}(\bsl {x})\bigr) +{\mathscr N}_\zeta ^{(j)}(\bsl {y})\del_{\bsl x}^{j-1}\bft {curl}_{\bsl x}\bsl {v}(\bsl {x})\Big \}\Bigr\vert_{\bsl{y}=\bsl{x}/\varepsilon}.
				\end{equation*}
			For each positive integer $K,$ consider the subset of the set of trial fields in (\ref{I2}) obtained by truncating the above expansion:
				\begin {multline}\label {set1}
					U_K:=\biggl\{ \,\bsl {u}(\bsl {x}, \boldsymbol {\zeta}): \bsl {u}(\bsl {x}, \boldsymbol {\zeta})=\bsl {v}(\bsl {x})+\sum_{j=1}^K\ep^j\Big \{ \del_\bsl{y}\bigl({\mathscr K}_\zeta ^{(j)}(\bsl {y})\del_{\bsl x}^j\bsl {v}(\bsl {x})\bigr)\\ 
					+ \del_{\bsl x}\bigl({\mathscr K}_\zeta ^{(j-1)}(\bsl {y})\del_{\bsl x}^{j-1}\bsl {v}(\bsl {x})\bigr) +{\mathscr N}_\zeta ^{(j)}(\bsl {y})\del_{\bsl x}^{j-1}\bft {curl}_{\bsl x}\bsl {v}(\bsl {x})\Big \}\ {\rm for\ some\ \bsl{v}}\biggr\}.
				\end {multline}
			Here $\bsl {v}$ runs over the set $\bigl[C^\infty_{\rm per}({\mathbb T})\bigr]^3\cap {\mathcal X}(\torus)$ of $\torus$-periodic, divergence-free, smooth vector fields with zero mean. 
			Consider a minimisation problem for the same functional (\ref {I2}) over the set $U_K.$ 
						Substituting test functions from (\ref {set1}) into the equation (\ref {I2}) yields 
				\begin {equation}\label {UK1}
					\bar {E}_{\ep }(\bsl {u},\bsl {f})=\int_\torus \bigg \{ \sum_{j=0}^K\sum_{k=0}^K\frac {1}{2}\ep^{j+k}\tilde {h}^{(j,k)}\del
					^j\curl\,
					\bsl {v}
					\del
					^k\curl\,
					\bsl {v}
					-\bsl {f}
					\cdot \bsl {v}
					\bigg \},
				\end {equation}
			where $\tilde {h}^{(j,k)}$ is a tensor of order $j+k+2$ defined by 
				\begin {equation}\label{hjk}
					\tilde {h}^{(j,k)}:=
					\Big \langle A \big ( \curl \,{\mathscr N}^{(j+1)}+{\mathscr M}^{(j+1)}\big ) \big ( \curl\,{\mathscr N}^{(k+1)}+{\mathscr M}^{(k+1)}\big ) \Big \rangle,
					\end{equation}
					or in the index notation:
					\[
					\tilde{h}^{(j,k)}_{i_1...i_{j+k+2}}=\Big \langle A_{st}\big ( \curl \,{\mathscr N}^{(j+1)}+{\mathscr M}^{(j+1)}\big )_{si_1...i_{j+1}}\big ( \curl\, {\mathscr N}^{(k+1)}+{\mathscr M}^{(k+1)}\big )_{ti_{j+2}...i_{j+k+2}}\Big \rangle.
					\]
			Whenever $\bsl{u}$ and $\bsl{v}$ are related via the expression in the definition of $U_K,$ we set 
			 \[
			 E_K(\bsl {v},\bsl {f},\ep):=\bar {E}_{\ep }(\bsl {u},\bsl {f}),
			 \]
			 so that 
			 \begin{equation}
				\min_{\bsl {v}
				}E_K(\bsl {v},\bsl {f},\ep)=\min_{\bsl {u}(\bsl {x}, \boldsymbol {\zeta})\in U_K}\bar {E}_\ep(\bsl {u},\bsl {f}),
			\label{var6}
			 \end{equation}
			 where $\bsl{v}\in\bigl[C_{\rm per}^\infty({\mathbb T})\bigr]^3\cap {\mathcal X}(\torus).$
%
			Clearly for all $\ep,$ $\bsl {f},$ the functional $E_K(\cdot,\bsl {f},\ep)$ is convex, 
			as a result of 
			$\boldsymbol{\zeta}$-averaging of the convex functional $E_{\ep,\zeta}$. 
			The next proposition, which is similar to \cite[Proposition 3]{bib13}, shows that the 
			minimiser $\bsl {v}_K$ in (\ref{var6}) is the best variational choice of a truncated approximation.
				\begin {Prop}\label {prop3}
					For $K\geq 2$ and all functions $\bsl {v}$ one has
						\begin{equation}
						E_K(\bsl {v},\bsl {f},\ep)\geq E_K(\bsl {v}_K,\bsl {f},\ep)\geq \bar {I}(\ep,\bsl {f}).
					\label{two_sided}
						\end{equation}
					Moreover, the estimate
						\begin {equation}\label {remainder1}
							E_K(\bsl {v}_K,\bsl {f},\ep)-\bar {I}(\ep,\bsl {f})\le c_K \ep^{2K}
						\end {equation}
					holds for some constant $c_K>0.$
				\end {Prop}
				\begin {proof}
				The inequalities (\ref{two_sided}) follow from the minimising property of 
					$\bsl {v}_K$ 
					and the fact that $\bar {I}_K(\ep,\bsl {f})\geq \bar {I}(\ep,\bsl {f}).$
					 To obtain the estimate (\ref{remainder1}), we substitute 
					 $\uep_\zeta(\bsl {x})=\bsl {u}^{(K)}_\zeta(\bsl {x},\ep)+\bsl {R}^{(K)}_\zeta(\bsl {x}, \bsl {f},\ep)$ into (\ref {var1}) and integrate by parts.
					 The remainder $\bsl{R}^{(K)}_\zeta$ satisfies bounds analogous to those proved for $\bsl{R}^{(K)}$ in Theorem \ref{thrm1},
					 hence
					        \[
						0\le E_\ep\bigl(\bsl {u}^{(K)}_\zeta,\bsl {f}\bigr)-I_\zeta(\ep,\bsl {f})\le c_K \ep^{2K}.
						\]
						Integrating the last inequality with respect to $\boldsymbol{\zeta}\in Q$ and using the minimising property of $\bsl{v}$ again, yields (\ref{remainder1}).
				\end {proof}
						Note the following alternative formula for $\tilde {h}^{(j,k)}.$ Using integration by parts, for $j=1,2,\dots, $ we get							$$\big \langle A \big ( \curl \ {\mathscr N}^{(j+1)}+{\mathscr M}^{(j+1)}\big ) \curl\ \bsl {\phi} \big \rangle =-\big \langle {\mathscr L}^{(j+1)}\bsl {\phi}\big \rangle,\ \ \ \ \ \ \ \ \forall \bsl {\phi}\in\bigl[C^\infty_{\rm per}(Q)\bigr]^3, 
						\ \ \langle \bsl {\phi}\rangle =0,$$
						hence
							\begin {equation}\label {hjk1}
								\tilde {h}^{(j,k)}=\big \langle A \big ( \curl \ {\mathscr N}^{(j+1)}+{\mathscr M}^{(j+1)}\big ) {\mathscr M}^{(k+1)}
								-{\mathscr L}^{(j+1)}{\mathscr N}^{(k+1)}\big \rangle,
							\end {equation}
						where
							\begin {equation*}
								\big({\mathscr L}^{(j+1)}{\mathscr N}^{(k+1)}\big)_{i_1\dots i_{j+k+2}}={\mathscr L}^{(j+1)}_{si_1\dots i_{j+1}}{\mathscr N}^{(k+1)}_{si_{j+2}\dots i_{j+k+2}}.
							\end {equation*}

		\subsection {Infinite-order variational homogenised equation}\label {sec125}
			The weak form of the Euler-Lagrange equation for the problem (\ref{var6}) reads
				\begin{equation}
				\label {EL2}
					\int_\torus \bigg \{ \frac {1}{2}\sum_{j,k=0}^K\ep^{j+k}\tilde {h}^{(j,k)}\Big ( \del
					^j\curl\,
					\boldsymbol {\varphi}\del
					^k\curl\,
					\bsl {v}_K+\del
					^j\curl\,
					\bsl {v}_K\del
					^k\curl\,
					\boldsymbol {\varphi}\Big )-\bsl {f}\cdot\boldsymbol {\varphi}\bigg \}
					=0\ \ \ \ \forall\,\boldsymbol{\varphi}\in\bigl[C^\infty_{\rm per}(\mathbb T)\bigr]^3.
				\end{equation}
			For $K=0$ this gives the identity
				$$\int_\torus \Bigl( \frac {1}{2}\tilde {h}^{(0,0)}\curl\,
				\bphi\cdot\curl\,
				\bsl {v}_0+\frac{1}{2}\tilde {h}^{(0,0)}\curl\,
				\bsl {v}_0\cdot\curl\,
				\boldsymbol {\varphi}-\bsl {f}\cdot  \boldsymbol {\varphi}\Bigr)
				=0\ \ \ \ \ \forall\,\boldsymbol{\varphi}\in\bigl[C^\infty_{\rm per}(\mathbb T)\bigr]^3,$$ or in the differential form
				\begin {equation*}
					\curl
					\bigl\{\tilde {h}^{(0,0)}\curl\,
					\bsl {v}_0\bigr\}=\bsl {f},
				\end {equation*}
where $\tilde {h}^{(0,0)}=\hat {h}^{(2)},$ as follows from (\ref {hjk}). 
Proceeding in a similar manner, 
			the Euler-Lagrange equation for the case $K=1$
			reads
				\begin{align}
					& \curl
					\bigl\{\tilde {h}^{(0,0)}\curl\,
					\bsl {v}_1\bigr\} +\ep\,\curl
					\bigl\{\tilde{\tilde{h}}^1\del
					\curl\,
					\bsl {v}_1\bigr\}+\ep^2\,\curl
					\bigl\{\tilde{h}^{(1,1)}\del
					\curl\,
					\bsl {v}_1\bigr\}=\bsl {f},\nonumber\\[0.6em]
					& \tilde {\tilde {h}}^{1}_{ijk}:=\frac {1}{2} \Big \{ \tilde {h}^{(0,1)}_{ijk}-\tilde {h}^{(0,1)}_{kji}-\tilde {h}^{(1,0)}_{jik} + \tilde {h}^{(1,0)}_{jki}\Big \}. \label{EL3a}
					\end{align}
			Consider the case when formally $K=\infty$ in 
			(\ref {EL2}) and denote by $\bsl {v}^{(\infty )}$ the corresponding minimiser. Notice that 
			\[
			\min_{\bsl{v}
			}E_\infty(\bsl{v}, \bsl{f}, \varepsilon)=\min_{\bsl{u}(\bsl{x},\zeta)} E_\varepsilon(\bsl{u}, \bsl{f})=\bar{I}(\varepsilon, \bsl{f}),
			\]
			where the minimiser $\bsl{u}(\bsl{x},\boldsymbol{\zeta})=\bsl{u}_\zeta^\varepsilon(\bsl{x}),$ and therefore $\bsl{v}^{(\infty)}
			=\bar{\bsl{u}}^\varepsilon.$
			 After a series of manipulations, the related Euler-Lagrange equation is written as
				\begin {equation}\label {EL4}
					\curl
					\bigl\{\tilde {h}^{(0,0)}\curl\,
					\bsl {v}^{(\infty )}\bigr\}+\sum_{n=1}^\infty \ep^n\curl
					\bigl\{\tilde {\tilde {h}}^{n}\del
					^{n}\curl\,
					\bsl {v}^{(\infty )}\bigr\}=\bsl {f},
				\end {equation}
			where ({\it cf.} (\ref{hjk}))
				\begin{equation}
				\label {EL5}
					\tilde {\tilde {h}}^{n}_{i_1\dots  i_{n+2}}:=\frac {1}{2}\sum_{\substack {j+k=n, \\ j,k\in \mathbb {N}\cup\{0\}}}\Big \{ (-1)^j\tilde {h}^{(j,k)}_{i_{j+1}i_2\dots i_ji_1i_{j+2}\dots i_{j+k+2}}
					+(-1)^k\tilde {h}^{(j,k)}_{i_{j+1}i_2\dots i_ji_{j+k+2}i_{j+2}\dots i_{j+k+1}i_1}\Big \}.
			      \end{equation}
			In Section \ref {sec126} we show that the ``asymptotic'' infinite-order homogenised equation (\ref {hom1}) coincides with the ``variational'' infinite-order homogenised equation (\ref {EL4}).

			The next result shows that the minimiser $\bsl {v}_K$ of $E_K(\bsl{v}, \bsl{f}, \varepsilon)$ is an approximation of order $O\bigl(\varepsilon^{2K}\bigr)$ to
			the infinite-order homogenised solution.
				\begin {Prop}\label {prop004}
					Let $K\geq 2$ be a positive integer. Then for any vector function $\bsl {f}$, there exists a constant $\hat {c}_K>0$ such that
						\begin {equation*}
							\bigl\Vert\bar{\bsl{u}}^\varepsilon
							-\bsl {v}_K\bigr\Vert_{[L^2(\torus)]^3}
							\leq \hat {c}_K\ep^{K}.
						\end {equation*}
				\end {Prop}
				\begin {proof}
				The proof follows the argument of \cite[Appendix D]{bib13}.
					Denote by $\bsl{u}_K^\zeta\in U_K$ the vector associated with $\bsl{v}_K$
					by (\ref {set1}), and let 
						$\bsl {R}_K^\zeta(\bsl {x}):=\bsl {u}^\ep_\zeta(\bsl {x})-\bsl {u}_K^{\zeta}(\bsl {x},\ep).$
					Notice first that
						$$E_K(\bsl {v}_K,\bsl {f},\ep)=\bar {E}_\ep\bigl(\bsl {u}_K^\zeta,\bsl {f}\bigr)=\int_Q E_{\ep,\zeta}\bigl(\bsl {u}_K^\zeta,\bsl {f}\bigr)\mathrm {d} \boldsymbol {\zeta}=\int_Q\int_\torus\bigg (\frac {1}{2}\hat{A}_\zeta^\varepsilon\,\curl \ \bsl {u}_K^\zeta\!\cdot \curl\,\bsl {u}_K^\zeta-\bsl {f}\cdot \bsl{u}_K^\zeta\bigg )\dx \mathrm {d} \boldsymbol {\zeta}$$
					\vskip-0.7cm
						\begin {multline*}
							=\int_Q\int_\torus\bigg (\frac {1}{2}\hat{A}_\zeta^\varepsilon\,\curl \ \bsl {u}^{\ep}_\zeta\cdot \curl \ \bsl {u}^{\ep}_\zeta-\frac {1}{2}\hat{A}_\zeta^\varepsilon\,\curl\,\bsl {u}^{\ep}_\zeta\cdot \curl\,\bsl {R}_K^\zeta\\ 
							-\frac {1}{2}\hat{A}_\zeta^\varepsilon\,\curl\,\bsl {R}_K^\zeta\cdot \curl \ \bsl {u}^{\ep}_\zeta+\frac {1}{2}\hat{A}_\zeta^\varepsilon\,\curl\,\bsl {R}_K^\zeta\cdot
							\curl\,\bsl {R}_K^\zeta-\bsl {f}\cdot \bsl {u}^\ep_{\zeta}+\bsl {f}\cdot \bsl {R}_K^\zeta\bigg )\dx \mathrm {d} \boldsymbol {\zeta}.
						\end {multline*}
					Since $\bsl {u}^\ep_\zeta$ solves (\ref {zeta1}), integrating by parts in the last expression we obtain
						\begin{equation}
						E_K(\bsl {v}_K,\bsl {f},\ep)=\bar {I}(\ep,\bsl {f})+\int_Q\int_\torus \frac {1}{2}\hat{A}_\zeta^\varepsilon\,\curl\,\bsl {R}_K^\zeta\cdot \curl\,\bsl {R}_K^\zeta\dx \mathrm {d} \boldsymbol {\zeta}.
						\label{onestar}
						\end{equation}
					The fact that the matrix $A_\zeta$ is positive-definite implies
						\begin{equation}
						\int_Q\int_\torus
						\hat{A}_\zeta^\varepsilon\,\curl\,\bsl {R}_K^\zeta\cdot \curl\,\bsl {R}_K^\zeta\dx \mathrm {d} \boldsymbol {\zeta}\geq
						\nu
						\int_Q \int_\torus\bigl|\curl \,\bsl {R}_K^\zeta\bigr|^2\dx \mathrm {d} \boldsymbol {\zeta}
						\ \ \ \ \ \ \ \ \ \ \ \ \ \ \ \ \ \ \ \ \ \ \ \ \ \ \ \ \ \ \ \ \ \ \ \ \ \ \ 
						\end{equation}
						\begin{equation}
						\ \ \ \ \ \ \ \ \ \ \ \ \ \ \ \ \ \ \ \ \ \ \ \ \ \ \ \ \ \ \ \geq 
					\nu
						\int_\torus \int_Q\bigl|\curl\,\bsl {R}_K^\zeta\bigr|^2\mathrm {d} \boldsymbol {\zeta}\dx\geq 
						\nu
						\int_\torus\bigl|\curl\,(\bar {\bsl {u}}^\ep-\bsl {v}_K)\bigr|^2.
						\label{twostar}
						\end{equation}
					Combining (\ref{onestar}), (\ref{twostar}), and (\ref{remainder1}) yields
						$$\int_\torus \bigl |\curl\,(\bar {\bsl {u}}^\ep-\bsl {v}_K)\bigr |^2
						\leq 2c_K\nu^{-1}\ep^{2K}.$$
					Finally, noticing that ${\rm div}\,\bar{\bsl {u}}^\ep={\rm div}\,\bsl{v}_K=0$ and using (\ref{Maxwell_inequality}), we obtain the desired result.
				\end {proof}


		\section {Tensor analysis of the infinite-order homogenised equations}\label {sec126}
		 Here we prove that the infinite-order homogenised equations (\ref {hom1}) and (\ref {EL4}) are equivalent. To this end, we introduce a symmetrisation operation, as follows.
		 \begin {Def}
		 \label{symmetrisation_definition}
					For $n\geq 1$, the  symmetrisation of a tensor $h_{ik_1\dots k_nj}$ of order $n+2$ is a tensor 
					$h_{i(k_1\dots k_n)j}$ of the same order
					defined by 
						\begin {equation*}
							h_{i(k_1\dots k_n)j}:=\frac {1}{n!}\sum_{(k_1,k_2,\dots,k_n)}h_{ik_1\dots k_nj},
						\end {equation*}
					where the summation is over all permutations of the indices $k_1,k_2,\dots,k_n$. 
				\end {Def}
				Notice that in the case of third-order tensors, the above operation leaves the tensor unchanged. In order to account for this special case, 
			we show separately that the third-order tensors $\hat {h}^{(3)}$ and $\tilde {\tilde {h}}^1$ coincide. 
			To this end, 
			notice first that 
				\begin {equation*}
				\tilde {h}^{(1,0)}_{ijk}=\Bigl\langle A_{st}\bigl(\curl\,
				{\mathscr N}^{(2)}+{\mathscr M}^{(2)}\bigr)_{sij}\bigl(\curl\,
				{\mathscr N}^{(1)}+I\bigr)_{tk}\Bigr\rangle\ \ \ \ \ \ \ \ \ \ \ \ \ \ \ \ \ \ \ \ \ \ \ \ \ \ \ \ \ \ \ \ 
				\end{equation*}
				\begin{equation*}
				\ \ \ \ \ \ \ \ \ \ \ \ \ \ \ \ \ \ \ \ \ \ \ \ \ \ \ \ \ \ \ \ =\Bigl\langle A_{st}\bigl(\curl\,
				{\mathscr N}^{(1)}+I\bigr)_{sk}\bigl(\curl\,
				{\mathscr N}^{(2)}+{\mathscr M}^{(2)}\bigr)_{tij}\Bigr\rangle=\tilde {h}^{(0,1)}_{kij}.
				\end {equation*}
			Further, making use of the definitions (\ref {Mij}) and (\ref {Lij}) and relation (\ref{hjk1}) yields
				\begin {equation*}
					\hat {h}^{(3)}_{ijk}=\Bigl\langle A_{is}\bigl(\curl\,
					{\mathscr N}^{(2)}+{\mathscr M}^{(2)}\bigr)_{sjk}\Bigr\rangle=\tilde {h}^{(0,1)}_{ijk}-\tilde {h}^{(1,0)}_{jik}.
				\end {equation*}
			Hence, by virtue of (\ref{homco1}) and (\ref {EL3a}), one has
				$$\bigl(\hat {h}^{(3)}-\tilde {\tilde {h}}^1\bigr)_{ijk}=\tilde {h}^{(0,1)}_{ijk}-\tilde {h}^{(1,0)}_{jik}-\frac {1}{2} \Big \{ \tilde {h}^{(0,1)}_{ijk}-\tilde {h}^{(0,1)}_{kji}-\tilde {h}^{(1,0)}_{jik}+\tilde {h}^{(1,0)}_{jki}\Big \}\ \ \ \ \ \ \ \ \ \ \ \ \ \ \ \ \ \ \ \ \ \ \ \ \ \ \ \ \ \ \ \ 
				$$
				$$
				\ \ \ \ \ \ \ \ \ \ \ \ \ \ \ \ \ \ \ \ \ \ \ \ \ \ \ \ \ \ \ \ =\tilde {h}^{(0,1)}_{ijk}-\tilde {h}^{(0,1)}_{kji}-\frac {1}{2} \Big \{ \tilde {h}^{(0,1)}_{ijk}-\tilde {h}^{(0,1)}_{kji}-\tilde {h}^{(0,1)}_{kji}+\tilde {h}^{(0,1)}_{ijk}\Big \} =0.$$

				
			\begin{Rem}
			It can be shown that for a third-order tensor $h=(h_{ijk})$ the identity
					$\curl\{h\del\curl\,\bsl {v}\}=\bsl {0}$ holds for all $\bsl{v}$ 
				if and only if	
				$$h_{ijk}= \left \{ \begin {array}{ll} a_1, &  ijk\in\{ 122,133\} , \\ a_2, & ijk\in\{ 221,331\} , \\ a_3, & ijk\in\{ 211,233\} , \\ a_4, & ijk\in\{ 112,332\} , \\ a_5, & ijk\in\{ 311,322\} , \\ a_6, & ijk\in\{ 113,223\} , \\ a_1+a_2, & ijk\in\{ 111\} , \\ a_3+a_4, & ijk\in\{ 						222\} , \\ a_5+a_6, & ijk\in\{ 333\} ,\\ 0, & \text {otherwise}, \end {array} \right .$$
			 for some constants $a_l,$ $l=1,... , 6.$ 
			 It follows from the above that in the case of the tensor 
				$h=\hat {h}^{(3)}-\tilde {\tilde {h}}^1,$ the constants $a_l$ vanish for all $l.$ 

			\end{Rem}
			Returning to the case of an arbitrary order, consider the expression for the tensors appearing in 
			the variational approach ({\it cf.} (\ref{EL5})):
				\begin{equation}
				\label {tensor100}
					\tilde {\tilde {h}}^{n}_{ik_1\dots  k_nj}=\frac {1}{2}\sum_{\substack {p+q=n, \\ p,q\in \mathbb {N}\cup\{0\}}}\Big \{ (-1)^p\tilde {h}^{(p,q)}_{k_pk_1\dots k_{p-1}ik_{p+1}\dots k_{p+q}j}+(-1)^q\tilde {h}^{(p,q)}_{k_pk_1\dots k_{p-1}jk_{p+1}\dots 								k_{p+q}i}\Big \} .
				\end{equation}
			Symmetrising the above expression with respect to $k_1,$ $k_2,$ ..., $k_n$  gives
				\begin{equation}
				\label {tensor101}
					\tilde {\tilde {h}}^{n}_{i(k_1\dots  k_n)j}=\frac {1}{2n!}\sum_{(k_1,\dots,k_n)}\sum_{\substack {p+q=n, \\ p,q\in \mathbb {N}\cup\{0\}}}\Big \{ (-1)^p\tilde {h}^{(p,q)}_{k_pk_1\dots k_{p-1}ik_{p+1}\dots k_{p+q}j}+(-1)^q\tilde {h}^{(p,q)}_{k_pk_1\dots 							k_{p-1}jk_{p+1}\dots k_{p+q}i}\Big \} .
		                \end{equation}
			The symmetry property
						$$
						\tilde {h}^{(p,q)}_{ik_1\dots k_pk_{p+1}\dots  k_{p+q}j}=\tilde {h}^{(q,p)}_{k_{p+1}\dots k_{p+q}jik_1\dots k_p}
						$$
		implies
		\begin{equation*}
							\frac {1}{n!}\sum_{(k_1,\dots ,k_n)}\sum_{\substack {p+q=n, \\ p,q\in \mathbb {N}\cup\{0\}}}(-1)^p\tilde {h}^{(p,q)}_{k_pk_1\dots k_{p-1}ik_{p+1}\dots  k_{p+q}j}=\frac {1}{n!}\sum_{(k_1,\dots ,k_n)}\sum_{\substack {p+q=n, \\ p,q\in \mathbb {N}\cup\{0\}}}(-1)^q\tilde {h}^{(p,q)}_{k_pk_1\dots k_{p-1}jk_{p+1}\dots  k_{p+q}i},
							\end{equation*}
					and hence		
				\begin {equation}
							\tilde {\tilde {h}}^{n}_{i(k_1\dots  k_n)j}=\frac {1}{n!}\sum_{(k_1,\dots,k_n)}\sum_{\substack {p+q=n, \\ p,q\in \mathbb {N}\cup\{0\}}}(-1)^p\tilde {h}^{(p,q)}_{k_pk_1\dots k_{p-1}ik_{p+1}\dots i_{p+q}j}.
							\label{after_symmetrisation}
						\end {equation}
						
			 As is shown above, the third-order tensor $\hat {h}^{(3)}$ is expressed as a combination of the 
			 tensors $\tilde {h}^{(0,1)}$ and $\tilde {h}^{(1,0)}.$  Similarly, the tensor $\hat {h}^{(n+2)}$ is a combination of the tensors $\tilde {h}^{(0,n)},\tilde {h}^{(1,n-1)},\dots,\tilde {h}^{(n,0)}$ 
			 as follows:
				\begin {multline*}
					\hat {h}^{(n+2)}_{ik_1\dots k_nj}=\tilde {h}^{(n,0)}_{k_1\dots k_nji}-\tilde {h}^{(n-1,1)}_{k_2\dots k_njk_1i}+\tilde {h}^{(n-2,2)}_{k_3\dots k_njk_2k_1i}-\tilde {h}^{(n-3,3)}_{k_4\dots k_njk_3k_2k_1i}+\dots\\[0.6em]
					-(-1)^n\tilde {h}^{(1,n-1)}_{k_{n}jk_{n-1}k_{n-2}\dots k_1i}+(-1)^n\tilde {h}^{(0,n)}_{jk_nk_{n-1}\dots k_1 i}.
				\end {multline*}
				The related argument makes use of the
			 equations (\ref {Mij})--(\ref {Lij}) and the relation (\ref {hjk1}). 
			Symmetrising the above expression and using symmetry properties of the tensors $\tilde {h}^{(p,q)}$, we obtain
				\begin {equation}
					\hat {h}^{(n+2)}_{i(k_1\dots k_n)j}=\frac {1}{n!}\sum_{(k_1,\dots,k_n)}\Big \{ \tilde {h}^{(n,0)}_{k_1\dots k_nji}-\tilde {h}^{(n-1,1)}_{k_2\dots k_njk_1i}+\dots +(-1)^n\tilde {h}^{(0,n)}_{jk_nk_{n-1}\dots k_1 i}\Big \}
				\end {equation}
				\begin {equation}\label {tensor105}
					=\frac {1}{n!}\sum_{(k_1,\dots,k_n)}\Big \{ \tilde {h}^{(0,n)}_{ik_1\dots k_nj}-\tilde {h}^{(1,n-1)}_{k_1ik_2\dots k_nj}+\dots +(-1)^n\tilde {h}^{(n,0)}_{k_nk_{n-1}\dots k_1 ij}\Big \}
					\end{equation}
					\begin{equation}
					=\frac {1}{n!}\sum_{(k_1,\dots,k_n)}\sum_{\substack {p+q=n, \\ p,q\in \mathbb {N}\cup\{0\}}}(-1)^p\tilde {h}^{(p,q)}_{k_pk_1\dots k_{p-1}ik_{p+1}\dots i_{p+q}j}
					=\tilde {\tilde {h}}^{n}_{i(k_1\dots  k_n)j}.
				\end {equation}
				Hence, the following result is established ({\it cf.} \cite[Section 3.3.2]{bib13} for the scalar case):
				\begin {Thrm}
					For the tensors  $\hat {h}^{(n+2)}$ and $\tilde {\tilde {h}}^n$ given by (\ref {homco1}) and (\ref {EL5}) respectively, one has
						$$\hat {h}^{(n+2)}_{i(k_1\dots k_n)j}=\tilde {\tilde {h}}^n_{i(k_1\dots k_n)j}.$$
					In particular, since each of the infinite-order homogenised equations (\ref {hom1}) and (\ref {EL4}) 
					is unchanged under this symmetrisation process, these equations coincide.  
				\end {Thrm}

	\section {Homogenised constitutive laws for the system of Maxwell equations}\label {sec13}

		\subsection {Maxwell equations and quasistatic approximation}\label {sec131}
						 Written in time-harmonic form, the system of Maxwell equations for the electric $\bsl{E}^{\omega,\ep}$ and magnetic $\bsl{H}^{\omega,\ep}$ fields and electric displacement $\bsl{D}^{\omega,\ep}$ and magnetic induction $\bsl{B}^{\omega,\ep}$  in an $\varepsilon$-periodic medium is
				\begin {align}
					\curl\,\bsl {E}^{\omega,\ep}  =-\text {i}\omega \bsl {B}^{\omega,\ep},\ \ \ \ \ \ \ & \curl\,\bsl {H}^{\omega,\ep} =\text {i}\omega\bsl {D}^{\omega,\ep}+\bsl {J}^\omega,
					\label {max11} 
					\\[0.2em]
					\dive\,\bsl {D}^{\omega,\ep} =0,
					\ \ \ \ \ \ \ \ & \dive\,\bsl {B}^{\omega,\ep}  =0,
					\label {max13} 
					\\[0.2em]
					\bsl {B}^{\omega,\ep}=\hat {\mu}^\ep\bsl {H}^{\omega,\ep},\ \ \ \ \ \ \ & \bsl {D}^{\omega,\ep} =\hat {\epsilon}^\ep\bsl {E}^{\omega,\ep},
					\label {max12}
				\end {align}
				where $\omega>0$ is the frequency, and we assume that $-\dive\,\bsl {J}^\omega=0.$ The coefficients $\mu,$ $\epsilon$ in (\ref{max12}) are the matrices of magnetic permeability and electric permittivity, which are subject to the same assumptions as the function $A$ at the beginning of Section \ref{sec11}.
			Denoting generically by $\bsl{A}$ the fields entering (\ref{max11})--(\ref{max12}), we substitute
			\begin {equation}
							\bsl {A}^{\omega,\ep}(\bsl {x})=\sum_{j=0}^\infty({\rm i}\omega)^j\bsl {A}^\ep_j(\bsl {x}),
							\ \ \ \ \ \ \ 
							\bsl{A}_j^\varepsilon(\bsl{x})\in{\mathbb R}\ \ \forall\bsl{x}\in{\mathbb T},\ \ j=0,1,2,...,
						\end {equation}
			and compare the terms with equal powers of the frequency $\omega.$
			 This procedure, which we refer to as the ``quasistatic approximation'', results in the leading-order equations
				\begin {equation}\label{quasist_system}
				\begin {cases}
					\curl\bigl\{\left (\hat {\mu}^\ep\right )^{-1} \curl \,\bsl {E}^\ep_1\bigr\}=-\bsl {J}_0,\ \ \ \curl\bigl\{\left (\hat {\epsilon}^\ep\right )^{-1} \curl \,\bsl {H}^\ep_1\bigr\}=\curl\bigl\{\left (\hat {\epsilon}^\ep\right )^{-1}\bsl {J}_1\bigl\}, \\[0.3em]
					\dive (\hat {\epsilon}^\ep\bsl {E}^\ep_1)=-\dive\,\bsl {J}_2,\ \ \ \ \ \ \ \ \ \ \ \ \ \ \ \dive (\hat {\mu}^\ep\bsl {H}^\ep_1)=0.
				\end {cases}
				\end {equation}
			The focus of Sections \ref{sec11} and \ref{sec12}
			is the first equation in the system (\ref{quasist_system}),
			where we set 
			$\bsl {E}_1^\ep=:\bsl{u}^\varepsilon$ with $-\bsl {J}_0=:\bsl{f}.$ Note also that under the assumption that the electric permittivity is given by the identity matrix, one has
				$\dive\,\bsl{u}^\varepsilon
				=0.$ 
				\begin {Rem}
					The form of the 
					equation for the magnetic field 
					is somewhat different 
					to that 
					for the electric field. 
					Indeed, the right-hand side of it depends on the permittivity matrix $\hat {\epsilon}^\ep$ 
					and hence depends on $\ep.$ A modified approach 
					that deals with this feature is described in Section \ref {sec134}.
				\end {Rem}

		\subsection {Energy considerations}\label {sec132}

			 The total energy of the electromagnetic field is the sum of two parts: the electrostatic energy and the magnetic energy. 
			 For the system of equations (\ref {max11})--(\ref {max12}), the corresponding expressions are given by 
				\begin {equation}\label {energy24}
					u^{\omega,\ep}_{\text {elec}}:=\frac {1}{2}\int_\torus\bsl {E}^{\omega,\ep}\cdot\overline{\bsl {D}^{\omega,\ep}},
					\ \ \ \ \ \ \ u^{\omega,\ep}_{\text {mag}}:=\frac {1}{2}\int_\torus\bsl {B}^{\omega,\ep}\cdot\overline{\bsl {H}^{\omega,\ep}},
				\end {equation}
			where the bar stands for complex conjugation. 
			Using equations (\ref{max12})--(\ref{max11}) in (\ref {energy24}), we obtain
	                  \begin {equation}\label {energy25}
					u^{\omega,\ep}_{\text {elec}}:=\frac {1}{2\omega^2}\int_\torus(\hat {\epsilon}^\ep)^{-1}\big (\curl\,\bsl {H}^{\omega,\ep}-\bsl {J}^\omega\big )\cdot \overline{\big ( \curl\ \bsl {H}^{\omega,\ep}-\bsl {J}^\omega\big )},
				\end {equation}
				\begin {equation}\label {energy26}
					u^{\omega,\ep}_{\text {mag}}:=\frac {1}{2\omega^2}\int_\torus \big (\hat {\mu}^\ep)^{-1}\curl\ \bsl {E}^{\omega,\ep}\cdot\overline{\curl\,\bsl {E}^{\omega,\ep}}.
				\end {equation}
			 Note that in the modified expression the electric energy depends on the magnetic field and that the magnetic energy depends on the electric field.
				From here on, considerations will be restricted to real vector fields and hence the complex conjugation notation will be dropped.

			Substituting a formal series in powers of $\omega$ for $\bsl {E}^{\omega,\ep},$ $\bsl {H}^{\omega,\ep},$ $\bsl {D}^{\omega,\ep},$ $\bsl {B}^{\omega,\ep}$ into (\ref{energy24})--(\ref{energy26}) and comparing the coefficients in front of $({\rm i}\omega)^j,$ $j=0,1,2,$ 
			yields:
				\begin{equation}
					\int_\torus (\hat {\mu}^\ep)^{-1}\curl\,\bsl {E}^{\ep}_0\cdot \curl\ \bsl {E}_0^{\ep}
					=0,\ \ \ \ \ \ \ 
					\int_\torus (\hat {\mu}^\ep)^{-1}\curl\,\bsl {E}^{\ep}_1\cdot 
					\curl\,\bsl {E}_1^{\ep}
					=\int_\torus\bsl {B}_0^\ep\cdot\bsl {H}_0^\ep.
				\label{first_two}
				\end{equation}
			Note that due to the fact  that $\curl\,\bsl {E}_0^\ep=\bsl {0}$, the first equation in (\ref{first_two}) holds automatically.

			
			In the next section we use the approach developed in Section \ref {sec11} and Section \ref {sec12} 
			to derive higher-order constitutive relations between the leading-order magnetic  $\bsl {H}_0^\ep$ and induction $\bsl {B}_0^\ep$ fields. 
			The asymptotic and variational approaches lead to two expressions for such a higher-order constitutive law, which 
			are shown to coincide, by a symmetrisation procedure in Section \ref {sec126}.

		\subsection {Infinite-order constitutive relations}\label {sec133}
			 In this section, we derive an expression for the leading-order term $\bsl {H}_0^\ep,$ 
			 given that the expansions for the magnetic field $\bsl {B}_0^\ep$ and the electric field $\bsl {E}_1^\ep$ are known.
		In the quasistatic approximation we obtain 
				\begin {equation*}
						\curl\bigl\{(\hat {\mu}^\ep)^{-1}\curl\,\bsl {E}_1^\ep\bigr\}=\bsl {f},\ \ \ \ \ 
						\bsl {E}_1^\ep\in {\mathcal X}(\torus),
				\end {equation*}
			where $\bsl {f}:=-\bsl {J}_0\in\bigl[C^\infty_{\rm per}({\mathbb T})\bigr]^3\cap {\mathcal X}(\torus).$ These equations 
			were analysed in Section \ref {sec11}, in particular, $\bsl {E}_1^\ep$ is written as a double series (\ref{A2})--(\ref{asymptoticv}) and the 
			infinite-order homogenised equation 
			(\ref{hom1}) is satisfied.
				

			We establish a higher-order constitutive relation between the magnetic field and magnetic induction
			by two approaches. First, we examine the family of problems considered in Section \ref {sec123}. In particular, we introduce a parameter $ \boldsymbol {\zeta}$ acting as a shift in the microscopic variable, which is reflected in the notation $\hat {\mu}^\ep_\zeta,$ 
			$\hat {\epsilon}^\ep_\zeta,$ $\bsl {H}_j^{\ep,\zeta},$ $\bsl {E}_j^{\ep,\zeta}.$
			Notice first that since $\hat {\epsilon}_\zeta^\ep$ is the identity matrix, the asymptotic expansion for the electric field displacement $\bsl {D}_1^{\ep,\zeta}$ is simply the asymptotic expansion for $\bsl {E}_1^{\ep,\zeta}$. Hence, averaging over $Q$ with respect to $ \boldsymbol {\zeta}$ yields
				$\bar {\bsl {D}}_1^\ep=\bar {\bsl {E}}_1^\ep=\bsl {v},$
			where the bar notation represents averaging over $Q$ with respect to $ \boldsymbol {\zeta}$. Further, note that 
				\begin {equation}\label {godknows1}
					-\bigl\langle \bsl {B}_0^{\ep,\zeta}\bigr\rangle_\zeta=\bigl\langle \curl\ \bsl {E}^{\ep,\zeta}_1\bigr\rangle_\zeta = \curl\,
					\bsl {v},
				\end {equation}
				and therefore
				\begin {align}
					\bar{\bsl {H}}_0^\ep:= & \bigl\langle \bsl {H}_0^{\ep,\zeta}\bigr\rangle_\zeta=-\bigl\langle (\hat {\mu}_\zeta^\ep)^{-1}\curl\,\bsl {E}_1^{\ep,\zeta}\bigr\rangle_\zeta 
					=-\sum_{j=0}^\infty\ep^j\hat {h}^{(j+2)}\del
					^j\curl\,
					\bsl {v}
					 \nonumber\\
					\stackrel{{\rm by\ }(\ref{uK_sum})}{=} & -\sum_{j=0}^\infty\ep^j\hat {h}^{(j+2)}\del
					^j\bigl\langle \curl\,\bsl {E}^{\ep,\zeta}_1\bigr\rangle_\zeta
				         \stackrel{{\rm by\ }(\ref{godknows1})}{=}  
					\sum_{j=0}^\infty\ep^j\hat {h}^{(j+2)}\del
					^j\bar {\bsl {B}}_0^\ep.\label{godknows2}
				\end {align}
			The resulting constitutive law is of the form $\bar {\bsl {H}}_0^\ep=(\mu^{\text {eff}})^{-1}\bar {\bsl {B}}_0^\ep$, where $(\mu^{\text {eff}})^{-1}$ is the inverse of the ``effective permeability operator". 

			In 
			\cite {bib13}, a higher-order stress-strain relation is derived via a variational argument by considering the corresponding elastic energy functional. 
			Applying the same approach to the magnetic energy functional in (\ref {first_two}), an alternative expression for the higher-order constitutive law is derived. 
			Results analogous to those of Section \ref {sec12} hold, in particular:
					 $$I(\ep ,\bsl {f}):=\frac {1}{2}\int_{\torus}(\hat {\mu}_\zeta^\varepsilon)^{-1}
					 \curl \,\bsl {E}_1^{\ep,\zeta} \cdot \curl \,\bsl {E}_1^{\ep,\zeta}
					 = -\frac {1}{2}\int_\torus \bsl {f}\cdot \bsl {E}_1^{\ep,\zeta}
					 \,\stackrel{\varepsilon\to0}{\sim}\,
					 \sum\limits_{k=0}^\infty \ep^kI_k(\bsl {f}),$$
					 $$I_k(\bsl {f})=-\frac {1}{2}\int_\torus \bsl {f}
												\cdot \bsl {v}_k,
												\ \ \ \ \ \ \ \ k=0,1,2,\dots,$$
						\begin {equation}\label {magen101}
							\bar{I}(\ep ,\bsl {f})=\frac {1}{2}\int_\torus \sum_{j,k=0}^\infty \ep^{j+k}\tilde {h}^{(j,k)}\del
							^j\curl\,
							\bsl {v}\del
							^k\curl\,
							\bsl {v},
						\end {equation}
					where $\tilde {h}^{(j,k)}$ is the tensor given by (\ref {EL5}).
			Integrating by parts in (\ref {magen101}) and using (\ref{godknows1}) yields, after an appropriate rearrangement of indices:
				\begin {equation}\label {magint1}
				\bar {u}^\ep_{\text {mag}}=\frac {1}{2}\int_\torus \bar {\bsl {\mathfrak {H}}}_0^\ep\cdot \bar {\bsl {B}}_0^\ep,
				 \ \ \ \ \ \ \ \ \bar {\bsl {\mathfrak {H}}}_0^\ep:=\sum_{n=0}^\infty\ep^n\sum_{\substack {j+k=n \\ j,k\in \nbb\cup\{0\}}}(-1)^k\bar {h}^{(j,k)}\del
				 ^n\bar {\bsl {B}}_0^\ep,
				\end {equation}
			where 
			$$\bar {h}^{(j,k)}_{i_1i_2\dots i_{j+k+2}}:=\tilde {h}^{(j,k)}_{i_{j+1}i_2\dots i_ji_{j+k+2}i_{j+2}\dots i_{j+k+1}i_1}.$$
			Using the symmetrisation procedure described in Section \ref {sec126} 
			it is confirmed that the expressions (\ref {godknows2}) and (\ref {magint1})  for the magnetic field  coincide to all orders. 


		\subsection {The magnetic field equation}\label {sec134}

			In this section we assume that the permeability $\hat {\mu}^\ep$ is the identity matrix and that the permittivity $\hat {\epsilon}^\ep$ is $Q$-periodic, symmetric and uniformly elliptic. We consider the quasistatic approximation of Section (\ref{sec132}).
			For a given vector function $\bsl {J}_1\in\bigl[C^\infty_{\rm per}(\torus)\bigr]^3\cap {\mathcal X}(\torus),$ we look for a solution to the problem
				\begin {equation}\label {hommag2}
						\curl\bigl\{(\hat {\epsilon}^\ep)^{-1}\curl\,\bsl {H}_1^\ep\bigr\}=\curl\bigl\{(\hat {\epsilon}^\ep)^{-1}\bsl {J}_1\bigr\},		\ \ \ \ \  \bsl {H}_1^\ep\in {\mathcal X}(\torus).
				\end {equation}
			The theory discussed in Section \ref {sec11} applies, subject to a modification to the asymptotic expansion (\ref {A2}): 
				\begin {multline}\label {mag100}
					\bsl {H}_1^\ep(\bsl {x})=\bsl {w}(\bsl {x},\ep)+\sum_{j=1}^{\infty}\ep^j\Bigl\{ \del_\bsl{y}\bigl({\mathscr S}^{(j)}(\bsl {y})\del_{\bsl x}^j\bsl {w}(\bsl {x},\ep)\bigr)+ \del_{\bsl x}\bigl({\mathscr S}^{(j-1)}(\bsl {y})\del_{\bsl x}^{j-1}\bsl {w}(\bsl {x},\ep)\bigr)\\[-0.7em] +{\mathscr T}^{(j)}(\bsl {y})\del_{\bsl x}^{j-1}\bigl(\bft {curl}_{\bsl x}\bsl {w}(\bsl {x},\ep)-\bsl {J}_1(\bsl {x})\bigr)\Big \}\Bigr\vert_{\bsl{y}=\bsl{x}/\varepsilon},
				\end {multline}
			where $\bsl {w}\in\bigl[C_{\text {per}}^\infty (\torus )\bigr]^3$ is a divergence-free vector field written as a series
				\begin {equation}\label {asymptoticw}
					\bsl {w}(\bsl {x},\ep)=
					\sum_{k=0}^\infty \ep^k\bsl {w}_k(\bsl {x}),\ \ \ \ \ \bsl{x}\in{\mathbb T},
				\end {equation}
			and ${\mathscr S}^{(j)},$ ${\mathscr T}^{(j)}$ are tensor fields of order $j+1,$ $j=0,1,\dots,$ ${\mathscr S}^{(0)}=\bsl {0}.$ 
			The equations of Section \ref {sec11} hold
			with ${\mathscr N}^{(j)}$ replaced by ${\mathscr T}^{(j)}$, ${\mathscr M}^{(j)}$ replaced by ${\mathscr R}^{(j)},$ and 
			$\hat {h}^{(j)}$ replaced by $\hat {k}^{(j)},$  where ${\mathscr R}^{(1)}=I,$ ${\mathscr R}^{(j+1)}\del_{\bsl x}=\curl_{\bsl x}{\mathscr T}^{(j)}.$ 
			In particular, 
				the homogenised equation of infinite-order has the form
				\begin{gather}
					\sum_{j=0}^{\infty}\ep^j\curl
					\bigl\{\hat {k}^{(j+2)}\del
					^j\curl\,
					\bsl {w}(\bsl {x},\ep)\bigr\}=\sum_{j=0}^\infty\ep^j\curl
					\bigl\{\hat {k}^{(j+2)}\del
					^j\bsl {J}_1(\bsl {x})\bigr\},\label {homog38}\\
					\ \ \ \ \ \ \ \ \ \ \ \ \ \ \ \ \ \ \ \ \ \ \ \ \ \ \ \ \ \ \ \ \ \ \ \			
					\hat {k}^{(j+2)}:=\bigl\langle \hat {\epsilon}^{-1} \bigl\{ \curl\,{\mathscr T}^{(j+1)}+{\mathscr R}^{(j+1)}\bigr\} \bigr\rangle.\nonumber
				\end {gather}
			A version of Theorem \ref {thrm1} on justification of the asymptotic procedure is proved, with similar error estimates. When considering the variational approach developed in Section \ref {sec12}, the argument is modified as follows. Consider the expression for the electric energy given by the first formula in (\ref {energy24}).
			Applying the ideas discussed in Section \ref {sec132}, we first apply the quasistatic approximation to the energy functional  and then write it 
			in terms of the magnetic field $\bsl {H}_1^{\omega,\ep},$ which 
			yields 
				\begin{equation}
				\int_\torus \bsl {E}_0^\ep\cdot \bsl {D}_0^\ep
				=
				\int_\torus(\hat {\epsilon}^\ep)^{-1} \big ( \curl\,\bsl {H}_1^\ep-\bsl {J}_1\big )\cdot \big ( \curl\,\bsl {H}_1^\ep-\bsl {J}_1\big ).
				\label{magenergy1}
				\end{equation}
			It is straightforward to see that (\ref {hommag2}) is the Euler-Lagrange equation for the minimisation problem 
				\begin {equation}\label {magfunc1}
					\min_{\bsl {H}}\frac {1}{2}\int_\torus(\hat {\epsilon}^\ep)^{-1} \big ( \curl\,\bsl {H}-\bsl {J}_1\big )\cdot \big ( \curl\,\bsl {H}-\bsl {J}_1\big )
					=:I(\ep,\bsl {J}_1),
				\end {equation}

			Analogous conclusions to those given in Section \ref {sec12} hold, in particular, an infinite-order homogenised equation is obtained by 
			by considering a minimisation problem over a restricted set of trial fields:
				\begin {equation}\label {VK1}
					\min_{\bsl{w}}\int_\torus \Big \{ \sum_{j=0}^K\sum_{l=0}^K\frac {1}{2}\ep^{j+l}\tilde {k}^{(j,l)}\del
					^j\bigl(\curl\,
					\bsl {w}
					-\bsl {J}_1
					\bigr)\del
					^l\bigl(\curl\,
					\bsl {w}
					-\bsl {J}_1
					\bigr)\Big \}
					=:\bar {E}_K(\bsl {w},\bsl {J}_1,\ep),
				\end {equation}
			where $\tilde {k}^{(j,l)}$ is a tensor of order $j+l+2$ given by
				\begin {equation}\label {kjl}
					\tilde {k}^{(j,l)}=\Bigl\langle\hat {\epsilon}^{-1} \big (\curl\,{\mathscr T}^{(j+1)}+{\mathscr R}^{(j+1)}\big ) \big ( \curl\,{\mathscr T}^{(l+1)}+{\mathscr R}^{(l+1)}\big ) \Bigr\rangle.
				\end {equation}
			Formally considering the case $K=\infty$ leads to the infinite-order homogenised equation 
				\begin {equation}\label {magEL4}
					\curl
					\bigl\{\tilde {k}^{(0,0)}\curl\,
					\bsl {w}^{(\infty )}\bigr\}+\sum_{n=1}^\infty \ep^n\curl
					\bigl\{\tilde {\tilde {k}}^{n}\del
					^{n}\curl\,
					\bsl {w}^{(\infty )}\bigr\}=\sum_{n=0}^\infty \ep^n\curl
					\bigl\{\tilde {\tilde {k}}^{n}\del
					^{n}\bsl {J}_1\bigr\},
				\end {equation}
			where
				\begin{equation}
				\label {magEL5}
					\tilde {\tilde {k}}^{n}_{i_1\dots  i_{n+2}}:=\frac {1}{2}\sum_{\substack {j+l=n, \\ j,l\in \mathbb {N}\cup\{0\}}}\Big \{ (-1)^j\tilde {k}^{(j,l)}_{i_{j+1}i_2\dots i_ji_1i_{j+2}\dots i_{j+l+2}}+(-1)^l\tilde {k}^{(j,l)}_{i_{j+1}i_2\dots i_ji_{j+l+2}i_{j+2}\dots i_{j+l+1}i_1}\Big \} .
		                 \end{equation}
			Applying the symmetrisation procedure of Section \ref {sec126}, it is checked that the equations (\ref{homog38}) and (\ref{magEL5}) coincide.

			 By analogy with a discussion in Section \ref {sec133}, we derive two versions of the infinite-order constitutive law between the electric field 
			 and electric displacement, 
			 using the Maxwell equations in combination with the infinite-order homogenised equation (\ref {homog38}) and by the analysis of electric energy:
				\begin {equation}\label {elecint1}
					\bar {\bsl {E}}_0^\ep=\sum_{j=0}^\infty\ep^j\hat {k}^{(j+2)}\del
					^j\bar {\bsl {D}}_0^\ep,\ \ \ \ \ \ \ \ 
					\bar {\bsl {\mathfrak {E}}}_0^\ep:=\sum_{n=0}^\infty\ep^n\sum_{\substack {j+l=n \\ j,l\in \nbb\cup\{0\}}}(-1)^l\bar {k}^{(j,l)}\del
					^n\bar {\bsl {D}}_0^\ep,
				\end {equation}
			where
				$$\bar {k}^{(j,l)}_{i_1i_2\dots i_{j+l+2}}:=\tilde {k}^{(j,l)}_{i_{j+1}i_2\dots i_ji_{j+l+2}i_{j+2}\dots i_{j+l+1}i_1}.$$

			\newpage
			\section*{Appendix A: Poincar\'{e}-type inequality}\label {sa1}
				In the proofs of Theorems \ref{unique1}, \ref{thrm1} 
				we use the following statement.
					\begin {Lem} 
					\label{lemma_Maxwell}
						For all $\bsl {v}\in L^2(\torus )^3$ such that $\curl\,\bsl {v}\in\bigl[L^2(\torus)\bigr]^3,$ $\dive\,\bsl {v}\in L^2(\torus ),$ and $\langle\bsl{v}\rangle_{\mathbb T}=0,$ the inequality
							\begin{equation}
							\| \bsl {v}\|^2_{[L^2(\torus )]^3}\leq 
							\vert\torus\vert\Bigl(\| \curl \,\bsl {v}\|^2_{[L^2(\torus )]^3}+\| \dive\,\bsl {v}\|^2_{L^2(\torus )}\Bigr)
							\label{Maxwell_inequality}
							\end{equation}
							holds.
					\end {Lem}
					\begin {proof}
						By virtue of $\bsl {v}\in\bigl[L^2(\torus )\bigr]^3$, $\langle\bsl{v}\rangle_{\mathbb T}=0,$ we write the Fourier series
							$$\bsl {v}(\bsl {x})=\sum_{\bsl {k}\in \mathbb {Z}^3\setminus\{0\}}\bsl {c}_{\bsl {k}}\rm {e}^{i\bsl {k}\cdot \bsl {x}},\ \ \ \ \ \ \ \ \ \ \ 
							\bsl {c}_{\bsl {k}}\in{\mathbb C}^3,\ \ \  \bsl{k}\in{\mathbb Z}^3.$$
						Noting that for $\bsl {k}\in{\mathbb Z}^3$ one has
							$\curl (\bsl {c}_{\bsl {k}}\rm {e}^{i\bsl {k}\cdot \bsl {x}})=i(\bsl {k}\times \bsl {c}_{\bsl {k}})\rm {e}^{i\bsl {k}\cdot \bsl {x}},$\ 
							$\dive(\bsl {c}_{\bsl {k}}\rm {e}^{i\bsl {k}\cdot \bsl {x}})=i(\bsl {k}\cdot \bsl {c}_{\bsl {k}})\rm {e}^{i\bsl {k}\cdot \bsl {x}},$
						it follows that
							\begin {align*}
						\| \bsl {v}\|^2_{[L^2(\torus)]^3} & =\vert\torus\vert\sum_{\bsl {k}\in \mathbb {Z}^3\setminus\{0\}}|\bsl {c}_{\bsl {k}}|^2\leq \vert\torus\vert\sum_{\bsl {k}\in \mathbb {Z}^3\setminus\{0\}}|\bsl {k}|^2|\bsl {c}_{\bsl {k}}|^2 \\
							&	=\vert\torus\vert\sum_{\bsl {k}\in \mathbb {Z}^3\setminus\{0\}}\bigl(|\bsl {k}\times \bsl {c}_{\bsl {k}}|^2+|\bsl {k}\cdot \bsl {c}_{\bsl {k}}|^2\bigr)=\vert\torus\vert\Bigl(\| \curl \,\bsl {v}\|^2_{[L^2(\torus)]^3}+\| \dive\,\bsl {v}\|^2_{L^2(\torus )}\Bigr),
							\end {align*}
						where we used the Parseval identity 
		and the equality $|\bsl {k}|^2|\bsl {c_k}|^2=|\bsl {k}\times \bsl {c_k}|^2+|\bsl {k}\cdot \bsl {c_k}|^2$.
					\end {proof}

\section*{Acknowledgement} This work was carried out under the financial support of
the Engineering and Physical Sciences Research Council (Grant EP/L018802/1 ``Mathematical foundations of metamaterials: homogenisation, dissipation and operator theory''). 


\begin{thebibliography}{99}
\bibitem{bib5} 
Bakhvalov, N., Panasenko, G., 1984. {\it  Homogenisation: Averaging Processes in Periodic Media,} Kluwer.
\bibitem{bib39}
Bensoussan, A., Lions, J. L., Papanicolaou, G., 1978. {\it Asymptotic Analysis for Periodic Structures,} North-Holland.
\bibitem{bib114}
Boutin, C., 1996. Microstructural effects in elastic composites. {\it Int. J. Solids and Structures} {\bf 33}(7), 1023--1051.
\bibitem{bib111}
Cherednichenko, K., 2001. {\it Higher-Order and Non-Local Effects in Homogenisation of Periodic Media}, PhD Thesis, University of Bath.
\bibitem{CherCooper}
Cherednichenko, K., Cooper, S., 2015. Homogenisation of the system of high-contrast Maxwell equations. {\it Mathematika} {\bf 61}(2), 475--500.
\bibitem{CherSm_2004}
 Cherednichenko, K. D., Smyshlyaev, V. P., 2004. On full two-scale expansion of the solutions of nonlinear periodic rapidly oscillating problems and higher-order homogenised variational problems. {\it Arch. Ration. Mech. Anal.} 
 {\bf174} (3), 385--442.
\bibitem{FH}
Fleck, N.A., Hutchinson, J.W., 1997. Strain gradient plasticity. {\it In: Hutchinson, J.W., Wu, T.Y. (Eds.),
Advances in Applied Mechanics,} {\bf 33} Academic Press, New York, 295--361.

\bibitem{bib24}
Gilbarg, D., Trudinger, N.S., 1998. {\it Elliptic Partial Differential Equations of Second Order}, Springer.

\bibitem{bib6}
Jikov, V. V., Kozlov, S. M., and Oleinik, O. A., 1994. {\it Homogenization of differential operators and 
integral functionals}, Springer. 
\bibitem{RamakrishnaGrzegorczyk}
Ramakrishna, S. A., Grzegorczyk, T.M., 2008. {\it Physics and Applications of Negative Refractive Index Materials,} CRC Press and SPIE Press.
\bibitem{bib25} 
Sanchez-Palencia, E., 1980. {\it Non-Homogeneous Media and Vibration Theory,} Springer.
\bibitem{bib13} 
Smyshlyaev, V. P., Cherednichenko, K. D., 2000. On rigorous derivation of strain gradient effects in the overall behaviour of periodic heterogeneous media, \textit{{J}. {M}ech. {P}hys. {S}olids} {\bf 48}, 1325--1357.
\bibitem{bib112}
Triantafyllidis, N., Bardenhagen, S., 1996. The influence of scale size on the stability of periodic solids
and the role of associated higher order gradient continuum models, {\it J. Mech. Phys. Solids} {\bf 44}(11), 1891--1928.
\bibitem{Zhikov2000}
Zhikov, V. V. 2000. On an extension of the method of two-scale convergence and its applications, {\it Sb. Math.}, 
{\bf 191}(7), 973--1014.

\end{thebibliography}

\end{document}